\documentclass[a4paper,oneside,11pt]{article}
\usepackage{a4wide, dsfont}
	\usepackage{epsfig}
\usepackage{graphicx}
\usepackage{multirow}
\usepackage[latin1]{inputenc}
\usepackage{amsmath}%
\usepackage{url}
\usepackage[round]{natbib}
\usepackage{color,soul} 
\usepackage{xcolor}
\usepackage{diagbox}

\usepackage{amssymb}%
\usepackage{amsthm}%
\usepackage{booktabs}%
\usepackage{dcolumn}%
\newcolumntype{.}{D{.}{.}{-1}}%
\newcolumntype{Z}{D{.}{.}{2}}%
\newcolumntype{T}{D{.}{.}{3}}
\newcolumntype{V}{D{.}{.}{4}}%
\theoremstyle{plain}
\newtheorem{theorem}{Theorem}%
\newtheorem{lemma}[theorem]{Lemma}%
\theoremstyle{definition}

\newcommand{\los}{\mathrm{L}}%
\newcommand{\losd}{\mathrm{L}_{\mathrm{Diff}}}%

\renewcommand\theta{\vartheta}

\newcommand{\E}{\ensuremath{\mathbb{E}}}%
\newcommand{\Pb}{\ensuremath{\mathbb{P}}}%
\newcommand{\R}{\ensuremath{\mathbb{R}}}%
\newcommand{\Aa}{\ensuremath{\mathcal{A}}}%
\newcommand{\F}{\ensuremath{\mathcal{F}}}%
\DeclareMathOperator{\Var}{Var}
\DeclareMathOperator{\Cov}{Cov}

\renewcommand\theta{\vartheta}

\begin{document}
%
%
\LARGE
%
{\title{Using Proxies to Improve Forecast Evaluation}
	\author{Hajo Holzmann\\
		\small{Fachbereich Mathematik und Informatik}  \\
		\small{Philipps-Universität Marburg} \\
		\small{holzmann@mathematik.uni-marburg.de}
		\and
		Bernhard Klar\footnote{Corresponding author. Bernhard Klar, Karlsruher Institut für Technologie, Institut für Stochastik, Englerstr. 2, 76131 Karlsruhe,  Germany} \\
		\\
		\small{Institut für Stochastik}  \\
		\small{Karlsruher Institut für Technologie (KIT)} \\
		\small{bernhard.klar@kit.edu} }
	}
	\maketitle
	%

%
%
%
\normalsize
\begin{abstract}
Comparative evaluation of forecasts of statistical functionals relies on comparing averaged losses of competing forecasts after the realization of the quantity $Y$, on which the functional is based, has been observed. Motivated by high-frequency finance, in this paper we investigate how proxies $\tilde Y$ for $Y$ - say volatility proxies - which are observed together with $Y$ can be utilized to improve forecast comparisons.
We extend previous results on robustness of loss functions for the mean to general moments and ratios of moments, and show in terms of the variance of differences of losses that using proxies will increase the power in comparative forecast tests. These results apply both to testing conditional as well as unconditional dominance. Finally, we numerically illustrate the theoretical results, both for simulated high-frequency data as well as for high-frequency log returns of several cryptocurrencies.
\end{abstract}
%


\bigskip

%
%
%
\section{Introduction}

Comparative evaluation of forecasts of statistical functionals is a standard issue in the realm of forecasting \citep{gneiting2011making}. It relies on comparing expected or averaged losses of competing forecasts,  
after the realization of the quantity $Y$, on which the functional is based, has been observed. The aim of this paper is to investigate how proxies $\tilde Y$ for $Y$, which are observed together with $Y$, can be utilized to improve forecast comparisons in the sense that they result in the same ordering but bring an increase in the power of comparative forecast tests.

The motivation comes mainly from high-frequency finance, where high-frequency data are routinely used to \textit{generate} forecasts - say of volatilities - also over moderate time horizons such as daily volatilities \citep{corsi2009simple}.
Our investigation shows how high-frequency data can be used to obtain \textit{sharper forecast evaluation}. In comparative forecast evaluation, this would mean that when comparing two forecasts of daily volatilities in terms of expected values of loss functions, the better forecast can be determined with higher power when using these high-frequency data in the process of forecast evaluation.

Our main point of departure was \cite{patton2011volatility}, who showed that using various noisy volatility proxies - e.g.~based on high frequency - is \textit{valid} in comparative forecast evaluation, that is, preserves the order of the expected losses. \cite{hansen2006consistent} have similar results, while \citet{laurent2013loss} provide a multivariate generalization of the characterization in \cite{patton2011volatility}, and \citet{koopman2005forecasting} illustrate the use of realized measures for forecast comparisons on various observed  high frequency data sets.
 We are interested in the comparison of different possibly misspecified forecasts, in which situation \citet{patton2020comparing} shows that the ranking of the forecasts may depend on the loss function. 

\smallskip

A very recent, related contribution is by \citet{hoga2022testing}, who focus on predicting the mean, and illustrate their methods for GDP forecasts.
Our contributions and their relation to the literature can be summarized as follows.
\begin{enumerate}
	\item We extend the analysis from \cite{patton2011volatility} and \cite{hansen2006consistent} about the validity when using various proxies from volatilities, that is, second moments, to general moments and beyond to ratios of moments. We use a concept corresponding to the notion of exact robustness from \citet{hoga2022testing}, who also assume that the proxy enters the loss difference in the same way as the original observation, and in this setting show that only the mean allows for exactly robust loss functions.
	\item We formally show in terms of the variance of differences of losses that using proxies will increase the power in comparative forecast testing by decreasing the variance of loss differences,  both when testing conditional as well as unconditional dominance, see \citet{nolde2017elicitability}) for these notions. \citet{hoga2022testing} have similar results for the mean and focus on conditional dominance testing.
	\item Finally, we illustrate the theoretical results for high-frequency data, both simulated as well as related to three cryptocurrencies, using the three-zone approach from \citet{fissler2016expected} and 	\citet{nolde2017elicitability}.
We show that the choice of the proxies, as well as the choice of the loss function, has a pronounced effect on the comparative evaluation of forecasts: using high-frequency data and the QLIKE loss substantially improves the forecast evaluation.
\end{enumerate}

The paper is structured as follows. In Section \ref{sec:constscore}, we start with a motivating example; we recall strictly consistent loss functions for statistical functionals, and introduce a dynamic framework for forecast evaluation. Section \ref{sec:moments} investigates the use of proxies to improve evaluations of forecasts of moments, and in Section \ref{sec:ratiomoments} this is further discussed and extended to ratios of moments. Section \ref{sec-sims} summarizes the results of a simulation study, where we consider comparing forecasts for second, third, and forth moments for GARCH-type time series. In Section \ref{sec-data}, we provide an illustration of our methods to predicting the volatility of log-returns of three cryptocurrencies, while a supplement contains additional numerical results.

%
%

\section{Motivation and Basic Concepts}\label{sec:constscore}
\subsection{Motivating examples} \label{sec-examples}
Let us first illustrate the use of different proxies and loss functions in Diebold-Mariano (DM) tests for equal forecast performance.
Consider the following stylized scenario, where the aim is to distinguish between two competing forecasts. Observations correspond to logarithmic returns, and we assume that the true data generating process is a simple GARCH(1,1) model. The total length of the time series is $T=1500$, and we consider rolling one-step-ahead forecasts of the conditional variance using a moving time window, with window length T/3, resulting in $n=1000$ forecasts.

There are four forecasters. The first one is lucky to use a GARCH(1,1) model for making predictions, forecasters 2, 3, and 4 use ARCH(1), ARCH(2), and ARCH(7) models, respectively.
Clearly, we expect that the predictions from forecaster 1 outperform, in some sense, the other predictions. Moreover, we would expect that the ARCH(7) model beats the ARCH(1) and ARCH(2) models since the former should be a better approximation to a GARCH(1,1) process. Let the forecasts of
the conditional variance of logarithmic returns $r_t$ for any pair of forecasters be denoted by $x_{1,t}$ and $x_{2,t}$. Then, our interest focuses on the null hypotheses
\begin{align*}
H_0: \quad \text{Forecast $x_{1,t}$ predicts at least as well as forecast $x_{2,t}$},
\end{align*}
and if $H_0$ is rejected, then $x_{2,t}$ is worse than $x_{1,t}$.
To decide for or against $H_0$, we use a DM test based on the loss differences
\begin{align*}
\Delta_n\bar{L} &= 1/n \sum_{t=1}^n L(x_{1,t},y_{t+1})-L(x_{2,t},y_{t+1}),
\end{align*}
where $L(x,y)$ is some loss function, and $y_{t+1}$ materializes at day $t+1$.
Under suitable conditions, the studentized test statistic $S$ has a limiting standard normal distribution, and $H_0$ is rejected for large values of $S$.

\smallskip
To evaluate the forecasts, the mean squared error loss $L(x,y)=(x-y)^2$, is used, together with the squared returns $\tilde{y}_t=r_t^2$ as an unbiased proxy for the true conditional variance of logarithmic returns.
The first line of the following table shows the values of the test statistic $S$ if the different predictions are compared to the GARCH(1,1) model. Even if the values are positive (hence, slightly favor the GARCH model), they are nowhere near statistically significant. The comparison of the ARCH(7) model with the other two ARCH models even results in values close to zero.

\begin{center}
\begin{tabular}{c|cccc}
\diagbox{$x_{2,t}$}{$x_{1,t}$} & GARCH(1,1) & ARCH(1) & ARCH(2) & ARCH(7)\\ \hline
 GARCH(1,1) &    -   &  0.788 & 1.010 & 0.984 \\
 ARCH(7)    & -0.984 & -0.193 & 0.152 & -
\end{tabular}
\end{center}

Next, assume that, besides the daily returns, also 5-min returns are available for predicting the next day's volatility. Thus, the squared returns are replaced by the realized variances $\tilde{y}_t= \sum_{i=1}^m r_{t,i}^2$, where $r_{t,i}$ are the intraday log returns. The outcomes of the DM tests for equal predictive performance, now using the realized variances as proxies, are as follows:

\begin{center}
\begin{tabular}{c|cccc}
\diagbox{$x_{2,t}$}{$x_{1,t}$} & GARCH(1,1) & ARCH(1) & ARCH(2) & ARCH(7)\\ \hline
 GARCH(1,1) &    -   & 3.976 & 3.924 & 3.506 \\
 ARCH(7)    & -3.506 & 2.474 & 2.352 &   -
\end{tabular}
\end{center}

In comparison with the first table, the values in the first line are much larger, being statistically significant even on the 0.01-level, and indicating the dominance of the prediction under the GARCH(1,1) model. The comparison of the ARCH(7) model with the other two ARCH models favors the ARCH(7) model, at least on the 0.05-level.

Finally, the evaluator decides to replace the MSE with the QLIKE loss function $\tilde{L}(x,y)=y/x-\log(y/x)-1$, and gets the following results of the corresponding DM tests.

\begin{center}
\begin{tabular}{c|cccc}
\diagbox{$x_{2,t}$}{$x_{1,t}$} & GARCH(1,1) & ARCH(1) & ARCH(2) & ARCH(7)\\ \hline
 GARCH(1,1) &  -  & 5.589 & 5.523 & 4.386 \\
 ARCH(7)    & -4.386 & 3.642 & 3.492 & -
\end{tabular}
\end{center}

Now, all entries are even larger in absolute terms, and the ARCH(7) model dominates the competing  ARCH models even on the 0.01-level.

Clearly, these results are based on a specific realization of the time series. However, a closer look at this example in Section \ref{sec-sims} reveals that this behavior is rather typical.

\smallskip
As a real data example, we consider log returns of the cryptocurrency Bitcoin (BTC). We use hourly observations from May 16, 2018 to October 27, 2021, with sample size 30264, which corresponds to 1261 days. All prices are closing values in U.S. dollars. Returns are estimated by taking logarithmic differences.
Figure \ref{fig:DM-BTC} shows the results of DM tests for the log returns for different competing models, where the results are depicted using the three-zone approach of \cite{fissler2016expected}. A result marked in green indicates that the model in the column outperforms the model in the row, a red mark indicates inferiority. The marking is yellow if there is no significant difference between the two models. Light green, green and dark green indicates significance at level 0.1, 0.05 and 0.01, respectively.
For example, in the upper left panel of Figure \ref{fig:DM-BTC}, the CGARCH model outperforms the ARCH(1) model, but none of the other models. Whereas there are only slight differences between the use of squared returns compared to high frequency data as proxies in case of MSE,
the power of the DM tests is considerably higher by using the high frequency proxy compared to squared returns with QLIKE loss. Moreover, there is an increase in power using the QLIKE loss function compared to MSE.
For more details, see Sections \ref{sec-DM} and \ref{sec-data}.

\begin{figure}
\centering
\includegraphics[width=0.9\textwidth,height=0.4\textheight]{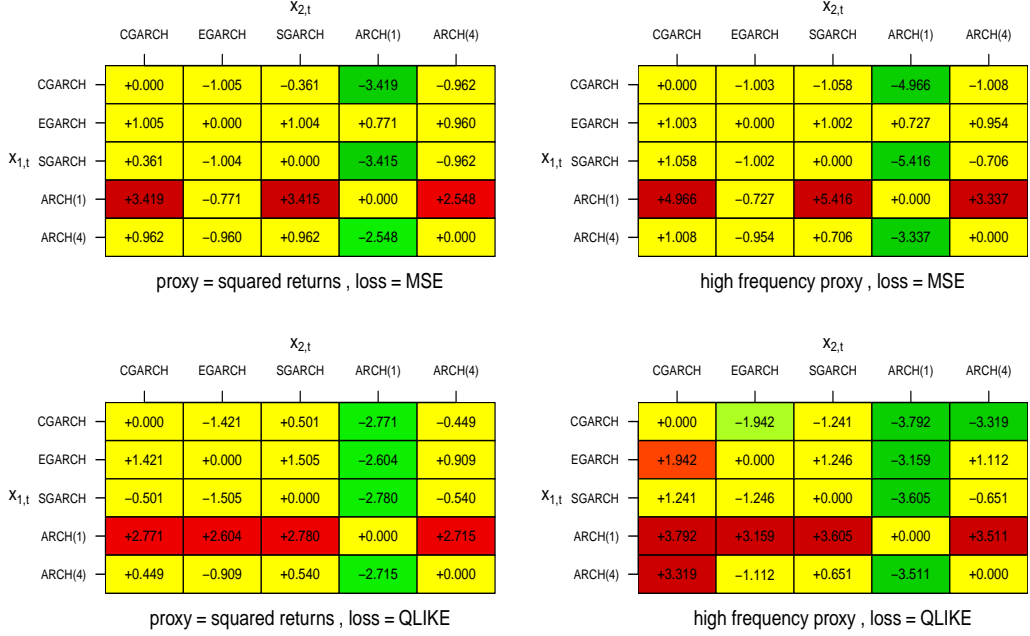}
\caption{Results of DM tests for BTC log returns from May\,16,\,2018 to October\,27,\,2021, left: squared returns, right: high frequency proxy.}
\label{fig:DM-BTC}
\end{figure}

\subsection{Loss functions and statistical functionals}
We start by recalling the concept of strictly consistent loss (or scoring) functions, see Gneiting (2011). Let
$\Theta$ be a class of distribution functions on a closed subset $D \subset \R$, which we identify with their associated probability distributions, and let
$T: \Theta \to \R$ be a (one-dimensional) \textit{statistical functional (or parameter)}.
%

A {\sl loss function (also scoring function)} is a measurable map
$\los: \R \times D \to [0, \infty)$.
It is interpreted as the loss if forecast $x$ is issued and $y$ materializes.
$\los$ is consistent for the functional $T$ relative to the class $\Theta$, if
\begin{equation}\label{eq:strctconsistent}
	\text{for all } x \in \R,\ F \in \Theta:\quad \E_F\big[\los(T(F),Y) \big] \leq \E_F \big[\los(x,Y )\big],
\end{equation}
and $\E_F$ indicates  that expectation is taken under the distribution $F$ for $Y$, and we assume that the relevant expected values are finite. Thus, the true functional $T(F)$ minimizes the expected loss under $F$.
If, in addition,
\[  \E_F\big[\los(T(F),Y) \big] = \E_F \big[\los(x,Y )\big] \quad \text{implies that} \quad  x = T(F),\]
then $\los$ is {\sl strictly consistent} for $T$. The functional $T$ is called {\sl elicitable} relative to the class $\Theta$ if it admits a strictly consistent loss function.
For several functionals such as moments, quantiles, and expectiles, \citet{gneiting2011making} characterizes all strictly consistent loss functions under some smoothness and normalization conditions.
%
See also \citet{steinwart2014elicitation}. 

When comparing two forecasts $x,x' \in \R$ for a given $F \in \Theta$ and hence parameter $T(F)$, we say that \textit{$x$ dominates $x'$ under $F$ for the loss $\los$} if the \textit{ difference of expected losses}
\begin{equation}\label{eq:scorediff}
	\E_F\big[\los(x,Y) \big] - \E_F \big[\los(x',Y )\big]< 0.
\end{equation}
From \eqref{eq:strctconsistent}, for a strictly consistent loss function, the true parameter $T(F)$ dominates any other forecast.

\subsection{Dynamic forcasting and comparative forecast evaluation}\label{sec:dynamic}

Now let us consider a forecasting situation. Forecasts are issued on the basis of certain information. Let $(\Omega, \Aa, \Pb)$ be a probability space, let $\F_t \subset \Aa$ be a sub-$\sigma$-algebra of $\Aa$, the information set at time $t$, on the basis of which the forecast is issued.
In finance, $\F_t$ can include returns (including high-frequency) up to time $t$ as well other covariates observed up to time $t$.

\smallskip

The aim is to predict the functional $T$, say the mean or the volatility, of the random variable $Y_{t+1}: \Omega \to \R$, which will be observed at time $t+1$ (say one day ahead). For example, this may be the return from $t$ to $t+1$ over one day.
More precisely, if
$ F_{ Y_{t+1}|\F_t}(\omega, \cdot)$
denotes the conditional distribution of $Y_{t+1}$ given $\F_t$, then the parameter of interest is
\[  T\big(Y_{t+1}|\F_t\big)(\omega) : = T\big( F_{ Y_{t+1}|\F_t}(\omega, \cdot)\big)\]
We note that
\begin{itemize}
	\item the forecast is based on the full information up to time $t$. Thus, even if $Y_{t+1}$ is a return over one day, for the forecast we use e.g.~high-frequency data up to time $t$ if these are included in $\F_t$,
	\item compared to such additional data, the observation $Y_{t+1}$ is of particular relevance since  the parameter  $T(Y_{t+1}|\F_t)$ is defined via its conditional distribution $F_{ Y_{t+1}|\F_t}(\omega, \cdot)$ given $\F_t$.
\end{itemize}
Thus, to generate the forecast, even if the time horizon for the forecast is, say, one day, it is standard to use high-frequency information contained in $\F_t$ up to time $t$.
Now, the issue in this paper is how to use additional information contained in $\F_{t+1}$, available at time $t+1$, to improve comparative forecast evaluation.

\smallskip

First let us recall the setting for comparative forecast evaluation based on $Y_{t+1}$.
A forecast at time $t$ is - in great generality - an $\F_t$-measurable random variable $Z_t$. Now, if $\los$ is a strictly consistent loss function for $T$, then compared to the true forecast $T(Y_{t+1}|\F_t)$, we have the following results \citep{holzmann2014role}:
\begin{equation}\label{eq:ineqfirst}
	\E\, \big[\los\big(T(Y_{t+1}|\F_t), Y_{t+1}\big) | \F_t\big](\omega) \leq \E\, \big[\los(Z_t, Y_{t+1}) | \F_t\big](\omega) \qquad \text{ for }\Pb-a.e.\ \omega \in \Omega,
\end{equation}
that is the conditional dominance of the true forecast $T(Y_{t+1}|\F_t)$ over any other generic forecast $Z_t$, and
\begin{equation}\label{eq:meanloss}
	\E\, \big[\los\big(T(Y_{t+1}|\F_t), Y_{t+1}\big)\big] \leq \E\, \big[\los(Z_t, Y_{t+1})\big],
\end{equation}
the unconditional dominance of $T(Y_{t+1}|\F_t)$ over $Z_t$. When comparing with the true forecast $T(Y_{t+1}|\F_t)$ (but not in general), these two concepts coincide: We have equality in (\ref{eq:ineqfirst}) or (\ref{eq:meanloss}) if and only if $T(Y_{t+1}|\F_t) = Z_t$ $\Pb$-almost surely.
Note that $Y_{t+1}$ is used in the comparisons \eqref{eq:ineqfirst} and \eqref{eq:meanloss} by default.

\smallskip
\begin{sloppypar}
We shall generally compare two potentially misspecified forecasts, that is general $\F_t$-measurable random variables $Z_t$  and $Z_t'$ none of which needs to coincide with $T(Y_{t+1}|\F_t)$. Then, by definition, $Z_t'$ \textit{conditionally dominates} $Z_t$ for the loss function $\los$ if
\end{sloppypar}
\begin{equation}\label{eq:ineqfirst1}
	\E\, \big[\los(Z_t', Y_{t+1}) | \F_t\big](\omega) \leq \E\, \big[\los(Z_t, Y_{t+1}) | \F_t\big](\omega) \qquad \text{ for }\Pb-a.e.\ \omega \in \Omega,
\end{equation}
with strict inequality on a set of positive probability, while $Z_t'$ \textit{unconditionally dominates} $Z_t$ for the loss function $\los$ if
\begin{equation}\label{eq:meanloss1}
	\E\, \big[\los(Z_t', Y_{t+1})\big] < \E\, \big[\los(Z_t, Y_{t+1})\big].
\end{equation}
In this more general setting, conditional dominance still implies unconditional dominance, but the converse is not true in general.

Now, we consider how additional information contained in $\F_{t+1}$ (apart from $Y_{t+1}$) may  be used for \textit{forecast evaluation}.
In the context of high frequency financial data, apart from using the high-frequency data to \textit{generate} forecasts over daily time horizons,  we shall investigate these high-frequency data to obtain \textit{sharper forecast evaluation}.

\section{Proxies when Comparing Forecasts of Moments}\label{sec:moments}

Suppose that $D = I$ is an interval, $h: I \to \R$ is a measurable function such that $\E_F[|h(Y)|] < \infty$ for all $F \in \Theta$.  Then a classical result by \citet{savage1971elicitation}, see also \citet{gneiting2011making}, characterizes the strictly consistent scoring functions for $T(F) = \E_F[h(Y)]$ in the form
\begin{equation}\label{eq:scoreingfct}
	\los(x,y) = \phi(y) - \phi(x) - \phi'(x)\, (h(y)-x), \qquad y,x \in I,
\end{equation}
where $\phi: I \to \R$ is a strictly convex function with subgradient $\phi'$  for which $\E_F[|\phi(X)|] < \infty$ for all $F \in \Theta$, and $\phi'$ denotes the derivative of $\phi$.

The functional $T(F) = \E_F[h(Y)]$ is called the \textit{$h$ - moment} of $Y$, or  \textit{generalized moment} or simply a moment of $Y$. The classical moments arise for $h(x) = x^n$ for integers $n$. In our applications, we shall focus on the cases $n =2,3$ and $n=4$.

First, we formulate the following lemma in the static framework.
\begin{lemma}\label{lem:momentdiff}
	Consider \eqref{eq:scoreingfct} and forecasts $x_1,x_2 \in \R$.
	\begin{enumerate}
		\item 	The loss difference,
		\begin{align}
			\los(x_1,y) - \los(x_2,y) & = \phi(x_2)\, (1-x_2)-\phi(x_1)\, (1-x_1)  + \big(\phi'(x_2) - \phi'(x_1)\big)\, h(y) \label{eq:lossdiffmoment} \\
			& = : \losd\big(x_1,x_2,h(y)\big) \nonumber
		\end{align}
		depends on $y$ only through $h(y)$.
		\item If $F \in \Theta$ and $\tilde Y$ is a random variable (with a given distribution) such that $\E_F[h(Y)] =\E[ \tilde Y] $ (the moment is the same), then
		\begin{equation}\label{eq:equalscorediffmoment}
			\E_F\big[\losd\big(x_1,x_2,h(Y)\big) \big] = \E\big[\losd\big(x_1,x_2,\tilde Y\big) \big].
		\end{equation}
		\item We have that
		\begin{equation}
			\Var_F\big(\losd\big(x_1,x_2,h(Y)\big) = \big(\phi'(x_2) - \phi'(x_1)\big)^2\, \Var_F\big(h(Y)\big).
		\end{equation}
		Consequently if in addition to ii) it holds that $\Var( \tilde Y) \leq \Var_F[h(Y)] $, then we have that
		\begin{equation}\label{eq:vardiffdiffmoment1}
			\Var\big(\losd\big(x_1,x_2,\tilde Y\big) \big) \leq \Var_F\big(\losd\big(x_1,x_2,h(Y)\big) \big) .
		\end{equation}
		
	\end{enumerate}
\end{lemma}
Here, $\tilde Y$ plays the role of the proxy that shall be used to improve forecast evaluation. Part (ii) shows that using $\tilde Y$ instead of $h(Y)$ is \textit{valid} if $\E_F[h(Y)] =\E[ \tilde Y] $ in the sense that dominance relations of forecasts are preserved when using $\tilde Y$, while \eqref{eq:vardiffdiffmoment1} shows that evaluation of score differences is actually sharper based on $\tilde Y$ instead of $h(Y)$ if $\Var( \tilde Y) \leq \Var_F[h(Y)] $.

\citet{hoga2022testing} call the equality of loss-differences in \eqref{eq:equalscorediffmoment} exact robustness. When assuming that the proxy $\tilde Y$ enters the loss-difference in the same fashion as $Y$, they show that exact robustness can only hold for strictly consistent scoring functions of the mean. In our more flexible approach, we cover general moments and also ratios of moments, see below.

\begin{proof}
	Part (i) is easily checked by inserting the loss function \eqref{eq:scoreingfct}.
	
	\smallskip
	
	Concerning part (ii),  inserting $\losd$ from \eqref{eq:lossdiffmoment}, we get by assumption
	\begin{align*}
		&	\E_F\big[\losd\big(x_1,x_2,h(Y)\big) \big] - \E\big[\losd\big(x_1,x_2,\tilde Y\big) \big]\\
		= & \big(\phi'(x_2) - \phi'(x_1)\big)\, \Big(\E_F\big[h(Y) \big] -  \E\big[\tilde Y \big]\Big) = 0
	\end{align*}		
	Part (iii) follows similarly easily.
	
\end{proof}

Now, let's turn to the dynamic setting described in Section \ref{sec:dynamic}.

\begin{theorem}[Forecast dominance testing]\label{th:dommoment}
	Consider forecasting the conditional moment $\E\, \big[h(Y_{t+1}) | \F_t\big]$, and suppose that $\tilde Y_{t+1}$ is $\F_{t+1}$-measurable with
	\begin{equation}\label{eq:condmoment}
		\E\, \big[\tilde Y_{t+1} | \F_t\big] = \E\, \big[h(Y_{t+1}) | \F_t\big] \qquad \text{ a.s.}
	\end{equation}
	\begin{enumerate}
		\item 	For the loss difference \eqref{eq:lossdiffmoment}, for any two  forecasts $Z_t$ and $Z_t'$ ($\F_t$-measurable random variables),
		\begin{equation}\label{eq:equalscorediffmoment2}
			\E\big[\losd\big(Z_t,Z_t',h(Y_{t+1})\big) | \F_t\big] = \E\big[\losd\big(Z_t,Z_t', \tilde Y_{t+1}\big) | \F_t\big],
		\end{equation}
		and hence in particular
		\begin{equation}\label{eq:equalscorediffmoment3}
			\E\big[\losd\big(Z_t,Z_t',h(Y_{t+1})\big) \big] = \E\big[\losd\big(Z_t,Z_t', \tilde Y_{t+1}\big) \big].
		\end{equation}
		Thus, both conditional as well as unconditional dominance are preserved when using $\tilde Y_{t+1}$ instead of $h(Y_{t+1})$ in the forecast comparison.
		\item If in addition to \eqref{eq:condmoment} we have that
		\begin{equation*}
			\Var( \tilde Y_{t+1} | \F_t) \leq \Var(h(Y_{t+1}) | \F_t),
		\end{equation*}
		then
		\begin{equation}\label{eq:varscorediffmoment3}
			\Var\big(\losd\big(Z_t,Z_t', \tilde Y_{t+1}\big) | \F_t\big) \leq 	\Var\big(\losd\big(Z_t,Z_t',h(Y_{t+1})\big) | \F_t\big)
		\end{equation}
		as well as
		\begin{equation}\label{eq:varscorediffmoment4}
			\Var\big(\losd\big(Z_t,Z_t', \tilde Y_{t+1}\big) \big) \leq 	\Var\big(\losd\big(Z_t,Z_t',h(Y_{t+1})\big) \big).
		\end{equation}
		
	\end{enumerate}
\end{theorem}
The second part of the theorem shows that a variance reduction is achieved both for testing conditional as well as unconditional dominance. 
\begin{proof}
	(i): \eqref{eq:equalscorediffmoment2} is \eqref{eq:equalscorediffmoment} in Lemma \ref{lem:momentdiff}, conditional on $\F_t$, while \eqref{eq:equalscorediffmoment3} follows from \eqref{eq:equalscorediffmoment2} by taking expected values.
	
	\smallskip
	
	(ii): \eqref{eq:varscorediffmoment3} is \eqref{eq:vardiffdiffmoment1} in Lemma \ref{lem:momentdiff}, (iii), conditional on $\F_t$. As for \eqref{eq:varscorediffmoment4}, we have that
	\begin{align*}
		\Var\big(\losd\big(Z_t,Z_t', \tilde Y_{t+1}\big) \big) & = \E\Big[ \Var\big(\losd\big(Z_t,Z_t', \tilde Y_{t+1}\big)| \F_t\big)\Big] + \Var\Big(\E\big[\losd\big(Z_t,Z_t', \tilde Y_{t+1}\big) | \F_t\big]\Big)
	\end{align*}
	Since by \eqref{eq:equalscorediffmoment2},
	\[ \Var\Big(\E\big[\losd\big(Z_t,Z_t', \tilde Y_{t+1}\big) | \F_t\big]\Big) = \Var\Big(\E\big[\losd\big(Z_t,Z_t', h( Y_{t+1})\big) | \F_t\big]\Big)\]
	the conclusion follows since
	\begin{align*}
		\E\Big[ \Var\big(\losd\big(Z_t,Z_t', \tilde Y_{t+1}\big)| \F_t\big)\Big] & = \E\Big[ \big(\phi'(Z_t') - \phi'(Z_t)\big)^2\,\Var\big(\tilde Y_{t+1}| \F_t\big)\Big]\\
		& \leq \E\Big[ \big(\phi'(Z_t') - \phi'(Z_t)\big)^2\,\Var\big(h( Y_{t+1})| \F_t\big)\Big]\\
		& = \E\Big[ \Var\big(\losd\big(Z_t,Z_t', h( Y_{t+1})\big)| \F_t\big)\Big].
	\end{align*}
\end{proof}

\section{Ratios of Moments and Further Parameters}\label{sec:ratiomoments}

Suppose that $D = I$ is an interval, $h: I \to \R$ and $s:I \to (0, \infty)$ are measurable functions such that $\E_F[|h(Y)|] < \infty$, $\E_F[s(Y)] < \infty$ for all $F \in \Theta$. The target parameter is
\[ T(F) = \frac{\E_F[h(Y)]}{\E_F[s(Y)]}.\]
\citet{gneiting2011making} shows that strictly consistent loss functions for $T(F)$ are of the form
\begin{equation}\label{eq:scoreingfct1}
	\los(x,y) = s(y)\big(\phi(y) - \phi(x)\big)  - \phi'(x)\, \big(h(y)-x\, s(y)\big)  - \phi'(y)\, \big(h(y)-y\, s(y)\big)
\end{equation}
($y,x \in I$), where it is additionally assumed that
\[ \E_F[|h(Y)| |\phi'(Y)|] < \infty,  \; \E_F[|s(Y)| |\phi(Y)|] < \infty, \; \E_F[|Y|\, |s(Y)| |\phi(Y)|] < \infty, \; F \in \Theta.\]
\begin{lemma}\label{lem:ratiomoment}
	Consider \eqref{eq:scoreingfct1} and forecasts $x_1,x_2 \in \R$.
	\begin{enumerate}
		\item 	The loss difference,
		\begin{align}
			\los(x_1,y) - \los(x_2,y) & =  \big(\phi'(x_2) - \phi'(x_1)\big)\, h(y)
			+ \big(x_1\,\phi'(x_1) - \phi(x_1) - x_2\,\phi'(x_2) + \phi(x_2)\big)\, s(y) \label{eq:lossdiffmoment1} \\
			& = : \losd\big(x_1,x_2,h(y), s(y)\big) \nonumber
		\end{align}
		depends on $y$ only through $h(y)$ and $s(y)$.
		\item \sloppy
If $F \in \Theta$, and $\tilde Y_1, \tilde Y_2$ are random variables (with given distributions) such that $\E_F[h(Y)] = \E[ \tilde Y_1] $ and $\E_F[s(Y)] =\E[ \tilde Y_2] $ (the moment is the same), then
		\begin{equation}\label{eq:equalscorediffmoment1}
			\E_F\big[\losd\big(x_1,x_2,h(Y), s(Y)\big) \big]
			= \E\big[\losd\big(x_1,x_2,\tilde Y_1, \tilde Y_2\big) \big].
		\end{equation}
		\item If in addition to (ii) we have that $\Var(a\, \tilde Y_1 + b\, \tilde Y_2) \leq \Var_F(a\, h(Y) + b\, s(Y)) $ for $a,b \in \R$ we have that
		\begin{equation}\label{eq:vardiffdiffmoment2}
			\Var\big(\losd\big(x_1,x_2,\tilde Y_1, \tilde Y_2\big) \big) \leq \Var_F\big(\losd\big(x_1,x_2,h(Y),s(Y)\big) \big) .
		\end{equation}
		
	\end{enumerate}
\end{lemma}
\begin{proof}
	The form \eqref{eq:lossdiffmoment1} of the loss difference follows directly from inserting \eqref{eq:scoreingfct1}. Then \eqref{eq:equalscorediffmoment1} and \eqref{eq:vardiffdiffmoment2} follow immediately from the form of the loss difference.
\end{proof}
Note that the result of the lemma does not contradict Theorem 1 in Hoga and Dimitriadis (2022), since they only allow a single proxy $\hat Y_t$ which enters the loss function in the same way as $Y_t$, whereas in the above lemma we require proxies  $\tilde Y_1$ and $\tilde Y_2$ for the two moments $h(Y)$ and $s(Y)$.

\begin{theorem}[Forecast dominance testing: Ratio of moments]\label{th:domratiomom}
	Consider forecasting the ratio of conditional moments $\E\, \big[h(Y_{t+1}) | \F_t\big] / \E\, \big[s(Y_{t+1}) | \F_t\big]$, and suppose that $\tilde Y_{t+1}^{(j)}$ are $\F_{t+1}$-measurable with
	\begin{equation}\label{eq:condmoment1}
		\E\, \big[\tilde Y_{t+1}^{(1)} | \F_t\big] = \E\, \big[h(Y_{t+1}) | \F_t\big], \qquad \E\, \big[\tilde Y_{t+1}^{(2)} | \F_t\big] = \E\, \big[s(Y_{t+1}) | \F_t\big] \qquad \text{ a.s.}
	\end{equation}
	\begin{enumerate}
		\item 	For the loss difference \eqref{eq:lossdiffmoment1}, for any two  forecasts $Z_t$ and $Z_t'$ ($\F_t$-measurable random variables),
		\begin{equation}\label{eq:equalscorediffmoment4}
			\E\big[\losd\big(Z_t,Z_t',h(Y_{t+1}), s(Y_{t+1})\big) | \F_t\big] = \E\big[\losd\big(Z_t,Z_t', \tilde Y_{t+1}^{(1)}, \tilde Y_{t+1}^{(2)}\big) | \F_t\big],
		\end{equation}
		and hence in particular
		\begin{equation}\label{eq:equalscorediffmoment5}
			\E\big[\losd\big(Z_t,Z_t',h(Y_{t+1}), s(Y_{t+1})\big) \big] = \E\big[\losd\big(Z_t,Z_t', \tilde Y_{t+1}^{(1)}, \tilde Y_{t+1}^{(2)}\big) \big].
		\end{equation}
		%
		%
		\item If in addition to \eqref{eq:condmoment1} we have that for all $\F_t$-measurable random variables $V,W$, we have that
		\begin{align*}
			& V^2\, \Var( \tilde Y_{t+1}^{(1)} | \F_t) + W^2\, \Var( \tilde Y_{t+1}^{(2)} | \F_t) + \, 2 \, V\, W \, \Cov(\tilde Y_{t+1}^{(1)},  \tilde Y_{t+1}^{(2)} | \F_t) \\
			& \qquad \leq V^2\, \Var( h(Y_{t+1}) | \F_t) + W^2\, \Var( s(Y_{t+1}) | \F_t) + \, 2 \, V\, W \, \Cov(h(Y_{t+1}),  s(Y_{t+1}) | \F_t),
		\end{align*}
		then
		\begin{equation}\label{eq:varscorediffmoment5}
			\Var\big(\losd\big(Z_t,Z_t', \tilde Y_{t+1}^{(1)}, \tilde Y_{t+1}^{(2)}\big) | \F_t\big) \leq 	\Var\big(\losd\big(Z_t,Z_t',h(Y_{t+1}), s(Y_{t+1})\big) | \F_t\big)
		\end{equation}
		as well as
		\begin{equation}\label{eq:varscorediffmoment6}
\Var\big(\losd\big(Z_t,Z_t', \tilde Y_{t+1}^{(1)}, \tilde Y_{t+1}^{(2)}\big)\big) \leq 	\Var\big(\losd\big(Z_t,Z_t',h(Y_{t+1}), s(Y_{t+1})\big)\big)
		\end{equation}
		
	\end{enumerate}
\end{theorem}
The proof is immediate from Lemma \ref{lem:ratiomoment}, and the final inequality \eqref{eq:varscorediffmoment6} follows as \eqref{eq:varscorediffmoment4} in Theorem \ref{th:dommoment}. The condition for a potential variance reduction in Theorem \ref{th:domratiomom}, (ii), is more restrictive than that from Theorem \ref{th:dommoment}, (ii), and apart from relating the variances of  $s(Y)$ and $h(Y)$ to those of the proxies, also involves conditional covariances.

\medskip

\bigskip

Theorem \ref{th:domratiomom} does not apply to measures such as skewness and kurtosis, which even for centered distributions are known not to allow for strictly consistent scoring functions. 
However, the revelation principle, Theorem 4 in \citet{gneiting2011making}, and the elicitability (existence of a strictly consistent scoring function) and hence joint elicitability of moments implies that for centered distributions, these measures are elicitable when considered together with the second moment. Roughly speaking, for the skewness this involves the two-dimensional parameter consisting of third and second moment, and for the kurtosis consisting of fourth and second moment.
The analysis of the corresponding loss-differences is then similar to that in Theorem \ref{th:domratiomom}.

\section{Diebold-Mariano testing} \label{sec-DM}

We briefly summarize the DM test \citep{diebold1995comparing} for forecast dominance, where we shall focus on unconditional dominance. For a discussion of conditional dominance testing, together with asymptotic theory and local power analysis see \citet{hoga2022testing}.

As proposed in \citet{fissler2016expected} and \cite{nolde2017elicitability}, in comparative backtesting, we are interested in the following null hypotheses
\begin{align*}
	H_0^-: \quad \text{Forecast $x_{1,t}$ predicts at least as well as $x_{2,t}$}, \\
	H_0^+: \quad \text{Forecast $x_{1,t}$ predicts at most as well as $x_{2,t}$}.
\end{align*}
The forecast $x_{2,t}$ is used as a benchmark. If the hypothesis $H_0^-$ is rejected, then $x_{2,t}$ is worse than $x_{1,t}$; if $H_0^+$ is rejected, $x_{1,t}$ is better than $x_{2,t}$. The error of the first kind for rejecting one of the two hypotheses, even though they are true, can be controlled by the level of significance. As in \cite{nolde2017elicitability}, we define
\begin{align*}
	\lambda &= \lim_{n\to\infty} \frac1n \sum_{t=1}^n \E[\los(x_{1,t},Y_{t+1})-\los(x_{2,t},Y_{t+1})]
	= \E[\los(x_{1,t},Y_{t+1})-\los(x_{2,t},Y_{t+1})]
\end{align*}
(assuming first-order stationarity). Then, dominance of $x_{1,t}$ over $x_{2,t}$ is equivalent to $\lambda\leq 0$, and $x_{1,t}$ predicts at most as well as $x_{2,t}$ if $\lambda\geq 0$. Therefore, the comparative backtesting hypotheses can be reformulated as
\begin{align*}
	H_0^-: \; \lambda\leq 0, \qquad H_0^+: \; \lambda\geq 0.
\end{align*}
Forecast equality can be tested with the so-called DM test \citep{diebold1995comparing, giacomini2006tests,diebold2015comparing}, which is based on normalized loss differences. Here, the test statistic is given by
\begin{align*}
	S  & =  \frac{\sqrt{n} \, \Delta_n\bar{L}}{\hat\tau},
\end{align*}
where $\Delta_n\bar{L}=1/n\sum_{t=1}^n L(x_{1,t},y_{t+1})-L(x_{2,t},y_{t+1})$ and $\hat\tau^2$ is an estimator of the long-run asymptotic variance of the loss differences. One possible choice for $\hat\tau^2$ is
\begin{align*} 
	\hat\tau^2  & = \left\{
	\begin{array}{ll}
		\hat\gamma_0 & \text{ if } h=1, \\
		\hat\gamma_0 + 2\sum_{j=1}^{h-1} \hat\gamma_j & \text{ if } h\geq 2,
	\end{array}
	\right.
\end{align*}
where $\hat\gamma_j$ denotes the lag $j$ sample autocovariance of the sequence of loss differences \citep{gneiting2011comparing,lerch2017forecaster}. Another possible choice \citep{diks2011likelihoodbased} is
$\hat\tau^2=\hat\gamma_0 + 2\sum_{j=1}^{J} (1-j/J)\hat\gamma_j$,
where $J$ is the largest integer less than or equal to $n^{1/4}$. As a compromise, we used $\hat\tau^2=\hat\gamma_0+2\hat\gamma_1$.
Under the null hypothesis of a vanishing expected loss difference and some further regularity conditions, the test statistic $S$ is asymptotically standard normally distributed.
Therefore, we obtain an asymptotic level-$\eta$ test of $H_0^+$ if we reject the null hypothesis when $S\leq \Phi^{-1}(\eta)$, and of $H_0^-$ if we reject the null hypothesis when $S\geq \Phi^{-1}(1-\eta)$.

To evaluate the tests for a fixed significance level $\eta\in(0,1)$, we use the following three-zone approach of \cite{fissler2016expected}.
If $H_0^-$ is rejected at level $\eta$, we conclude that the forecast $x_{1,t}$ is worse than $x_{2,t}$, and we mark the result in red; similarly, if $H_0^+$ is rejected at level $\eta$, forecast $x_{1,t}$ is better than $x_{2,t}$, and we mark the result in green. Finally, if neither $H_0^-$ nor $H_0^+$ can be rejected, the marking is yellow.

\section{Simulations} \label{sec-sims}

In this section,  we report some results of an extensive simulation study. Additional results of the simulations are contained in Section 1 of the Supplementary Material \citep{holzmann2022proxy}. First, in Section \ref{sec-2-mom} we investigate proxies for the volatility, and then in Sections \ref{sec:highermoments} and \ref{sec:apARCH} turn to higher-order moments.

\subsection{Squared returns and realized variance} \label{sec-2-mom}

In the first two sections, as data generating process (DGP) for the log returns we use a  GARCH(1,1) process defined by
\begin{align*}
\sigma_t^2 = a_0+a_1r_{t-1}^2+b\sigma_{t-1}^2, \quad r_t=\sigma_t \varepsilon_{t},
\end{align*}
where additionally
\begin{align}
\varepsilon_{t}&=\sum_{i=1}^m \varepsilon_{t,i}, \quad
\varepsilon_{t,i} ={\cal N}(0,1/m), \ i=1,\ldots,m,
\end{align}
and all $\varepsilon_{t,i}$ independent. Assuming that $\sigma_t$ is constant on $(t-1,t]$, observed intraday returns are given by $r_{t,i}= \sigma_t\varepsilon_{t,i}$.
We use $a_0=0.02, a_1=0.08, b=0.85$, $m=100$ and $m=13$; the first is a typical range using 5-min returns, the latter corresponds to the use of half-hour returns at the New York Stock Exchange (NYSE).
While this is certainly an oversimplified model for actual high-frequency data, it serves to illustrate the use of proxies, and is similar to the setting used in \citet[Section 2.2]{patton2011volatility}. Real high - frequency data for log-returns of cryptocurrencies are analysed in Section \ref{sec-data}.

As the total length of the time series, we take $T=1500$ and $T=6000$.
For these two time series, we generate rolling one-step-ahead forecasts of the conditional variance using a moving time window, with window length $T/3$, refitted every 10 time steps, for GARCH(1,1), ARCH(1), ARCH(2), and ARCH(7) models. Hence, for $T=1500$, the fit is based on 500 values, and the DM tests use 1000 forecasts of each model.
All computations are done in R \citep{rcore2021r} using the R packages \verb+rugarch+ \citep{ghalanos2020rugarch} and \verb+fGarch+ \citep{wuertz2020fgarch}.

To stabilize the results, the following figures show the means of the results of 500 replications for Figures \ref{fig1} and \ref{fig2}, and of 50 replications for all figures after Figure \ref{fig2}, of the DM test.
All figures use the three-zone approach described in the previous section with the following modification: we simultaneously show rejection of $H_0^-$ at level $0.1,0.05$ and $0.01$ by marking in light red, red and dark red, respectively. Marking in light green, green and dark green signals rejection of $H_0^+$ at the three levels.
Besides the forecasts from the different (G)ARCH models, we show the result for the optimal forecast, given by the true conditional volatilities. Each figure shows four plot matrices: in the left (right) column, the squared returns (realized variances) are used as proxies. In the upper row, the loss function is the mean squared error $L(x,y)=(x-y)^2$, whereas in the lower row, the QLIKE loss function $\tilde{L}(x,y)=y/x-\log(y/x)-1$ is used. These loss functions correspond to the choice of $\phi(y)=y^2$ and $\phi(y)=-\log(y)$ in eq. (\ref{eq:scoreingfct}).
The power of the DM test and even the ranking of competing forecasts may depend on the particular choice of the strictly consistent loss function in case of misspecified forecasts. See \citet{patton2020comparing} and \citet{ehm2016quantiles} for a discussion.
The QLIKE loss  for the mean was proposed in \citet{patton2011volatility} as a $0$ - homogeneous alternative to the more standard squared loss. It requires fewer moment assumptions, and a favorable performance for volatility forecasting was found in various empirical studies including \citet{patton2009evaluating}.

\bigskip
Figure \ref{fig1} shows the results of DM tests under normal innovations, with $T=1500$ and $m=100$. The left panels show results for the squared returns $r_t^2$, the right panels for the realized variances $RV_t=\sum_{i=1}^m r_{t,i}^2$.

\begin{figure}
\centering
\includegraphics[width=0.9\textwidth,height=0.4\textheight]
{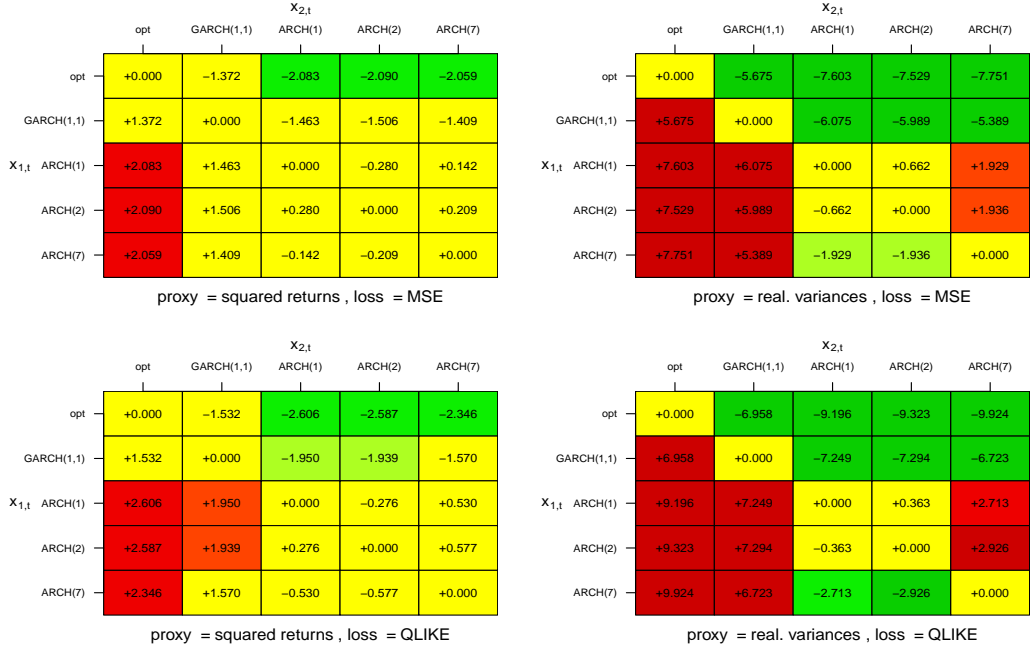}
\caption{Results of DM tests (based on 500 replications), normal distribution, left: squared returns, right: realized variances, $m=100, T=1500$.}
\label{fig1}
\end{figure}

Let us first discuss the results shown in the lower-left panel, i.e. for the QLIKE loss and using the squared returns as proxies. The second value in the left column, +1.532, is the average value of the DM test statistic comparing the forecast $x_{1,t}$ from the GARCH(1,1) model with the optimal forecast $x_{2,t}$, the true conditional volatilities. The positive value hints at the superiority of $x_{2,t}$, but the value is not statistically significant. The results are significant when comparing the ARCH(1), ARCH(2) and ARCH(7) model with the optimal forecast; here, the red color indicates a significant rejection of $H_0^-$ at the 0.05-level.
The light red entries in the second column show that forecasts from the ARCH(1) and ARCH(2) models are worse than the forecasts from the GARCH(1,1) model (which is the true data generating process) at level 0.1, but not on the 0.05-level. The forecast from GARCH(7) is not significantly worse than the GARCH(1,1).

Now, let's turn to the lower-right panel with the realized variances as proxies. Here, all corresponding entries are marked in dark red, signalizing the rejection of $H_0^-$ at level 0.01 in all cases. Hence, the power of the DM test is clearly higher by using realized variances compared to squared returns. Looking at the upper row, we see that the results for the MSE are similar from a qualitative point of view. However, there are fewer statistically significant entries compared to the QLIKE loss function. Hence, the latter allows for sharper forecast evaluation in this example.

\smallskip
Figure \ref{fig2} shows results from the same setting as Fig. \ref{fig1} apart from that we use $m=13$, corresponding to the use of half-hourly returns, instead of $m=100$. Hence, the results in the two left panels are the same as in Fig. \ref{fig1} (up to simulation error). The right panels are similar as in Fig. \ref{fig1}, as well. A closer look shows that all entries have smaller absolute values, showing the decreasing power in differentiating forecasts.

\begin{figure}
\centering
\includegraphics[width=0.9\textwidth,height=0.4\textheight]
{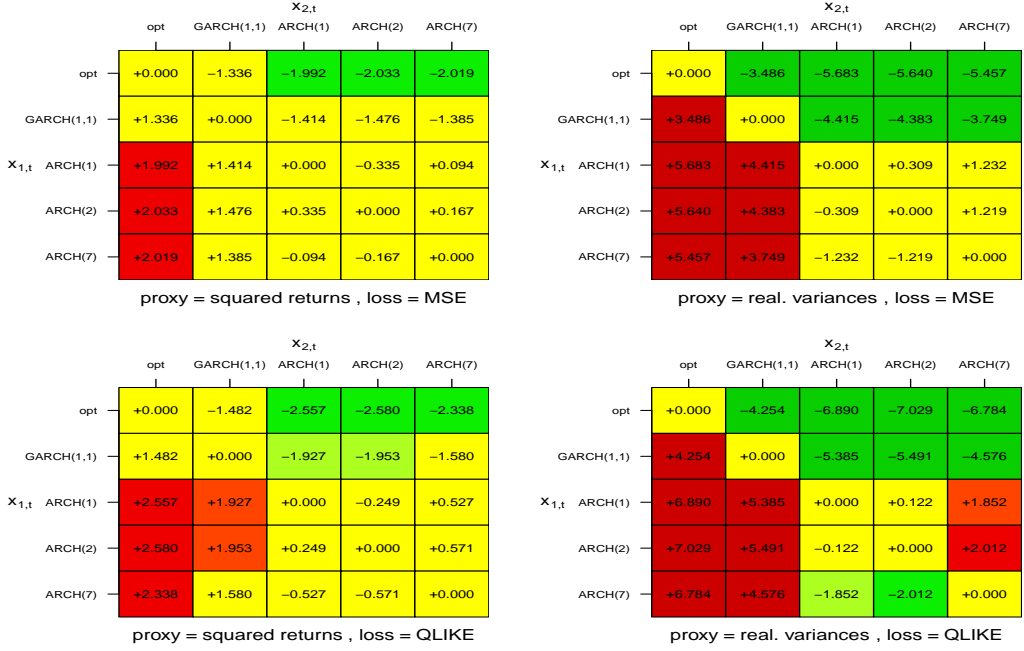}
\caption{Results of DM tests (based on 500 replications), normal distribution, left: squared returns, right: realized variances, $m=13, T=1500$.}
\label{fig2}
\end{figure}

\smallskip
We also replaced the normal distribution of the intraday innovations by centered skewed and long-tailed distributions.
For this, we used the normal inverse gaussian distribution $nig(\mu,\delta,\alpha,\beta)$ \citep{barndorff1997nig} with parameters
\begin{align*}
\alpha&=2, \quad \beta=1, \quad \gamma=\sqrt{\alpha^2-\beta^2}, \quad \delta=\gamma^3/\alpha^2/m, \quad \mu=-\delta\beta/\gamma.
\end{align*}
This results in $\E[\varepsilon_{t,i}]=0$ and $Var(\varepsilon_{t,i})=1/m$.
Since the class of nig distributions with fixed shape parameters $\alpha$ and $\beta$ is closed under convolution, the distribution of $\varepsilon_{t}=\sum_{i=1}^m \varepsilon_{t,i}$ is given by $nig(m\mu,m\delta,\alpha,\beta)$, with $\E[\varepsilon_{t}]=0, Var(\varepsilon_{t})=1$, \ $\E[\varepsilon_{t}^3]=1$, and $\E[\varepsilon_{t}^4]=17/3$.

Fig. \ref{fig5} shows the results of the DM tests as in Fig. \ref{fig1}, i.e. for $m=100$, using nig instead of normally distributed innovations. Again, the results are qualitatively comparable to the results in Fig. \ref{fig1}, but the absolute values of the entries are generally smaller. Hence, the change in the distribution of the innovations has a negative effect on the power of the test. Note that this decrease of power is larger for the realized variances than for the squared returns. This can be explained by the fact that the skewness and kurtosis of the daily innovations are rather modest with values of 1 and 5.67, whereas the skewness and kurtosis of the intraday innovations are 10 and  269.8, respectively.

\begin{figure}
\centering
\includegraphics[width=0.9\textwidth,height=0.4\textheight]
{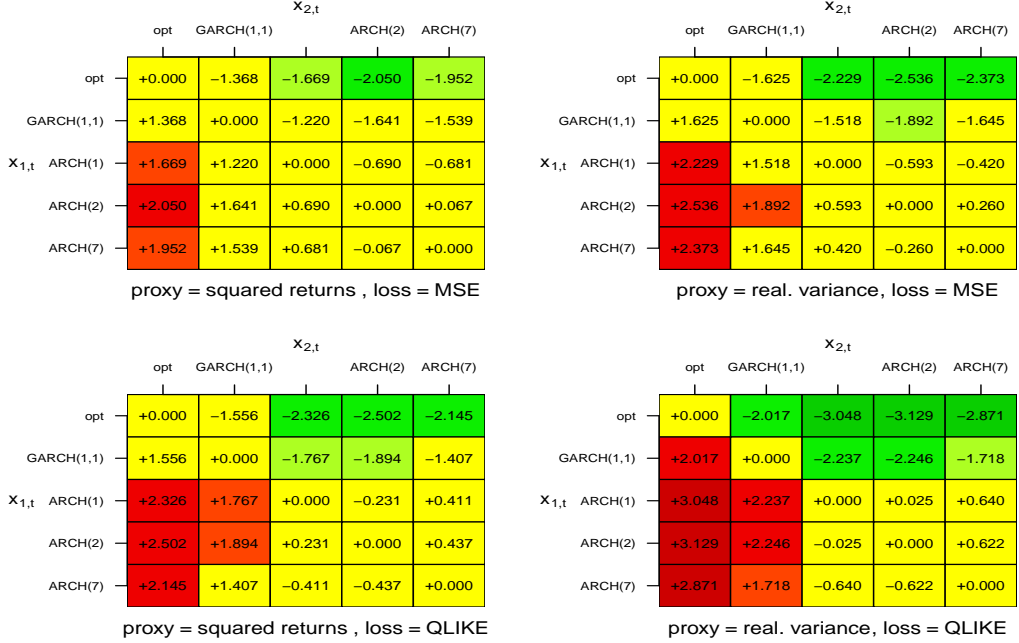}
\caption{Results of DM tests (based on 50 replications), nig distribution, left: squared returns, right: realized variances, m=100, T=1500.}
\label{fig5}
\end{figure}

\subsection{Higher moments}\label{sec:highermoments}

The use of realized higher moments, skewness, and kurtosis to estimate and forecast returns has become quite standard in the literature. For example,
\cite{neuberger2012realized} analyzed realized skewness and showed that high-frequency data can be used to provide more efficient estimates of the skewness in price changes over a period.
\cite{amaya2015realized} constructed measures of realized daily skewness and kurtosis based on intraday returns, and analyzed moment-based portfolios.
Recently, \cite{shen2018surprising} discussed the explanatory power of higher realized moments.

\subsubsection{Third moment} \label{sec-3-mom}

Assuming $\E[r_t|\F_{t-1}]=0$, we are interested in the conditional third moment $\rho_t=\E[r_t^3|\F_{t-1}]$. Possible proxies for $\rho_t$ are the cubed return $r_t^3$ and the realized third moment $RM(3)_t=\sum_{i=1}^m r_{t,i}^3$.
We use the GARCH(1,1) model of subsection \ref{sec-2-mom}, with the normal inverse Gaussian distribution for the innovations. Under this model, we obtain
\begin{align*}
\rho_t &= \E\left[r_t^3|\F_{t-1}\right]
 = \E\left[\sigma_t^3 \varepsilon_t^3|\F_{t-1}\right]
 = \sigma_t^3 \, \E[\varepsilon_t^3],  \\
\E\left[\varepsilon_t^3\right]
&= \E\left[ \left(\sum_{i=1}^m \varepsilon_{t,i}\right)^3 \right] \\
&= \E\left[ \sum_{i} \varepsilon_{t,i}^3  + 3 \sum_{i<j} \varepsilon_{t,i}^2  \varepsilon_{t,j} + 6 \sum_{i<j<k} \varepsilon_{t,i}\varepsilon_{t,j}\varepsilon_{t,k}  \right]
  = \sum_{i=1}^m \E\left[ \varepsilon_{t,i}^3 \right].
\end{align*}
Since
\begin{align*}
  \E\left[RM(3)_t|\F_{t-1}\right]
  &= \sigma_t^3 \sum_{i=1}^m \E\left[\varepsilon_{t,i}^3\right]
  = \sigma_t^3 \, \E[\varepsilon_t^3],
\end{align*}
$RM(3)_t$ is an unbiased estimator of $\rho_t$.
As forecast of $\rho_t$, we use $\tilde\sigma_t^3 \, \E[\varepsilon_t^3]$, where $\tilde\sigma_t$ denotes the one-step ahead forecast of $\sigma_t$ from the different (G)ARCH models.

%
%
%

\smallskip
Figure \ref{fig9} shows the results of DM tests under
$nig(\mu, \delta, 2, 1)$ innovations, with $\mu,\delta$ chosen such that
$\E[\varepsilon_{t,i}]=0, Var(\varepsilon_{t,i})=m^{-1/2}$.
Skewness and kurtosis of the intraday innovations are $3.61$ and 37.67, respectively, compared to the values 1 and 5.67 of the daily innovations.
Here, total length of the simulated time series is $T=6000$, and we use $m=13$, i.e. half-hourly returns. The left panels show the results for the cubed returns $r_t^3$, the right panels for the realized third moment  $RM(3)_t=\sum_{i=1}^m r_{t,i}^3$.

At first glance, the results seem to be rather different from the corresponding ones for the volatility, since the number of significant entries is much lower (cp. Fig. \ref{fig2}). But they go in the same direction: use of the realized moments increases the power of the DM test when the optimal forecast competes against the other models, or when the true data generating process is compared with ARCH models.

We have also used $T=1500$ in the simulations; the results (not shown) go in the same direction, but none of the values is statistically significant, even at the 0.1-level.

\begin{figure}
\centering
\includegraphics[width=0.9\textwidth,height=0.4\textheight]
{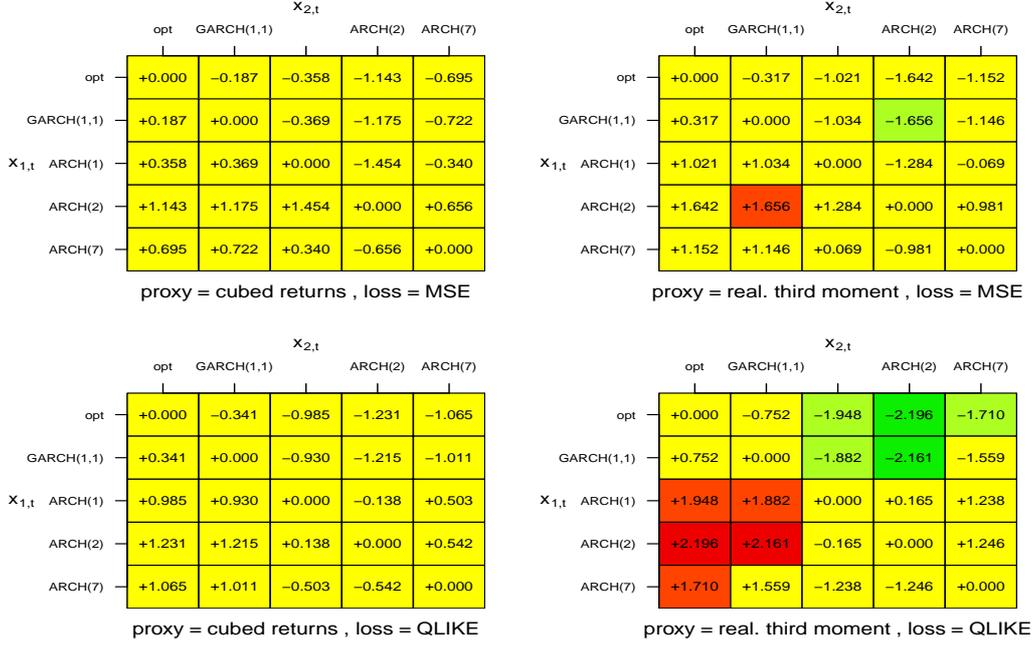}
\caption{Results of DM tests (based on 50 replications), nig-distribution, left: cubed returns, right: realized 3rd moment, m=13, T=6000.}
\label{fig9}
\end{figure}


\subsubsection{Fourth moment} \label{sec-4-mom}

Here, we are interested in the conditional fourth moment
$\tau_t = \E\left[r_t^4|\F_{t-1}\right]$.
Again, we use the GARCH(1,1) model as in subsection \ref{sec-2-mom}, and obtain
\begin{align*}
\E\left[\varepsilon_t^4\right]
 &= \sum_{i} \E[\varepsilon_{t,i}^4]
    + 6 \sum_{i<j} \E\left[ \varepsilon_{t,i}^2 \varepsilon_{t,j}^2\right]
 = \sum_{i} \E[\varepsilon_{t,i}^4]
    + 3 m(m-1) \left( \E[\varepsilon_{t,1}^2] \right)^2.
\end{align*}
Hence, unbiased proxies for $\tau_t$ are $r_t^4$ and the realized corrected fourth moment
\begin{align*}
cRM(4)_t &= \sum_{i=1}^m r_{t,i}^4 + 6 \sum_{i<j} r_{t,i}^2 r_{t,j}^2.
\end{align*}
As forecast of $\tau_t$, we use $\tilde\sigma_t^4 \, \E[\varepsilon_t^4]$.


\smallskip
The left and right panels of Fig. \ref{fig10} show the results of the DM tests, using $r_t^4$ and the realized corrected fourth moment as proxies, respectively. The innovations are normally distributed; further, $T=1500$ and $m=13$. The general picture resembles strongly the results of the volatility forecasts in Fig. \ref{fig2}, and all conclusions also apply here, even though the actual entries are a bit smaller.

\begin{figure}
\centering
\includegraphics[width=0.9\textwidth,height=0.4\textheight]
{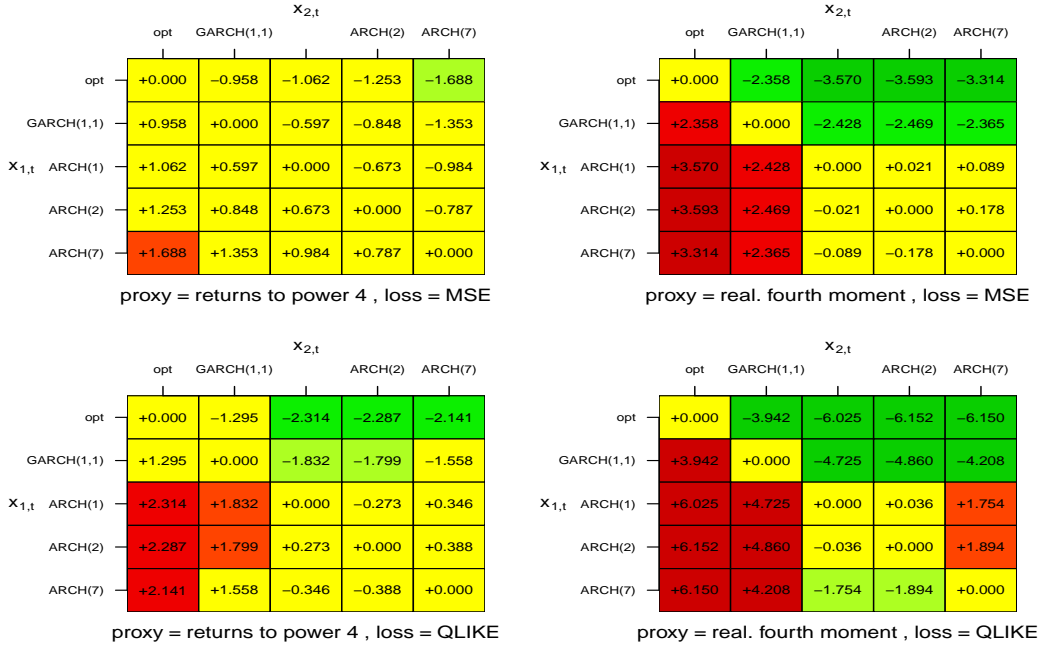}
\caption{Results of DM tests (based on 50 replications), normal distribution, left: returns to the power 4, right: realized corrected 4th moment, m=13, T=1500.}
\label{fig10}
\end{figure}

\smallskip
When replacing the normal by the nig innovations, the power of the DM test decreases strongly (cf. Fig. A.2 in the Supplementary Material \citep{holzmann2022proxy}). On the other hand, the entries are somewhat larger as in forecasting the third moment (with $T=1500$). Here, at least a few values are significant on the 0.1-level.

%

\subsection{An apARCH model for the fourth moment}\label{sec:apARCH}

Instead of modeling the volatility, and computing higher moments under this process, it is also possible to use suitable models for higher moments directly.
\citet{harvey1999autoregressive,harvey2000conditional}, for example, considered autoregressive model for conditional skewness.
\cite{lambert2002modeling} used the asymmetric power (G)ARCH or APARCH model of \cite{ding1993aparch} to describe dynamics in skewed location-scale distributions.
\cite{brooks2005autoregressive} used both separate and joint GARCH models for conditional variance and conditional kurtosis, whereas
\cite{lau2015simple} modeled (standardized) realized moments by an exponentially weighted moving average.

Hence, in this section, we model the fourth moment directly by an asymmetric power ARCH (apARCH) process \citep{ding1993aparch}. Specifically, the log returns follow an apARCH(1,1) model with $\delta=4$
\begin{align*}
\sigma_t^4 = \omega+\alpha r_{t-1}^4+\beta\sigma_{t-1}^4, \quad
r_t=\sigma_t \varepsilon_{t},
\end{align*}
where $\varepsilon_{t}=\sum_{i=1}^m \varepsilon_{t,i}$,
$\varepsilon_{t,i} ={\cal N}(0,1/m)$ for $i=1,\ldots,m$, and all $\varepsilon_{t,i}$ are independent. Assuming again that $\sigma_t$ is constant on $(t-1,t]$, intraday returns are given by $r_{t,i}= \sigma_t\varepsilon_{t,i}$.
We use $\omega=0.02, \alpha=0.08, \beta=0.75$ such that the unconditional variance is
\begin{align*}
 \sigma^2 &=
 \left( \frac{\omega}{1- E(\varepsilon_1^4)\alpha-\beta} \right)^{2/\delta}
 = \sqrt{2}.
\end{align*}
Further, $m=100$ and $T=1500$.
As in the last section, unbiased proxies for $\tau_t = \E\left[r_t^4|\F_{t-1}\right]$ are $r_t^4$ and $cRM(4)_t$.
As forecast of $\tau_t$, we use $\tilde\sigma_t^4 \, \E[\varepsilon_t^4]$, where $\tilde\sigma_t^4$ denotes the one-step ahead forecast of $\sigma_t^4$ from the different apARCH models, namely apARCH(1,1), apARCH(1), apARCH(2), and apARCH(3).

\smallskip
The left and right panels of Figure \ref{fig13} show the results of the DM tests for the apARCH process with exponent 4, with $r_t^4$ and the realized corrected fourth moment, respectively, as proxies.

The visual comparison of the upper-left and lower-right panels is striking: in the latter, each result is significant, whereas the former shows no significant entries. Hence, using high-frequency data and a suitable loss function results in a highly improved forecast evaluation. Generally, the results are quite similar to the results for the fourth moment based on the GARCH process in Fig. \ref{fig10}.

\begin{figure}
\centering
\includegraphics[width=0.9\textwidth,height=0.4\textheight]
{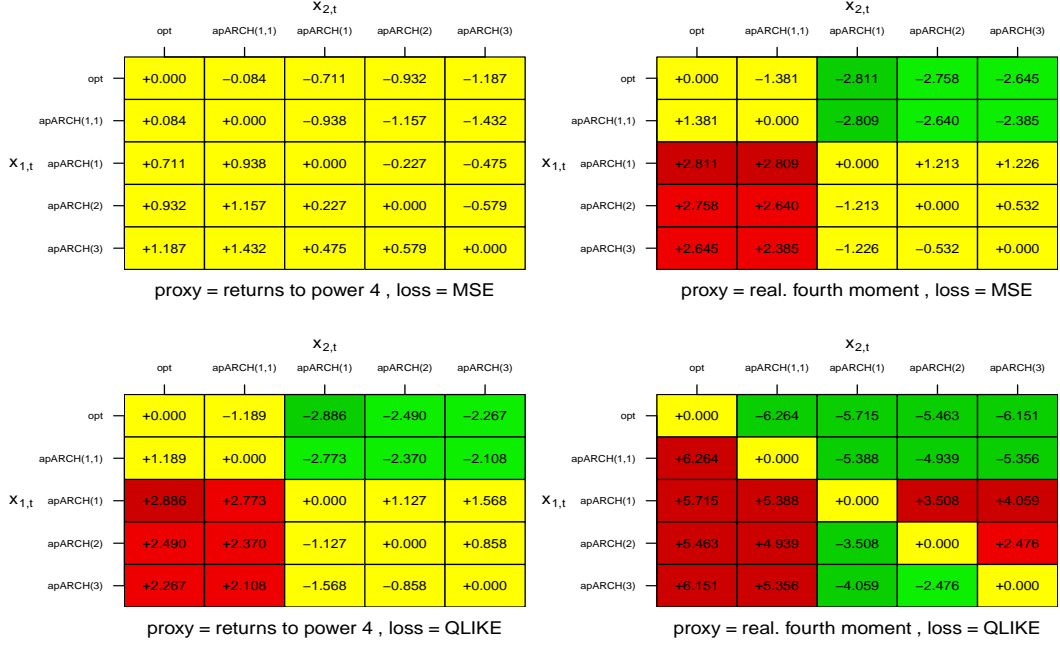}
\caption{Results of DM tests (based on 50 replications), apARCH process with exponent 4, normal distribution, left: returns to the power 4, right: realized corrected 4th moment, m=100, T=1500.}
\label{fig13}
\end{figure}

%
%

\smallskip
To sum up the results of the simulations, it has become obvious that using high-frequency data for the proxies improves the forecast evaluation in each example. In most cases, the effect is substantial. There is also an effect of the choice of the loss function: the power of the DM test improves when using the QLIKE loss compared to the MSE loss function.

\section{Log returns of cryptocurrencies} \label{sec-data}

Many GARCH and GARCH-type models have been used for modelling and predicting the volatility of cryptocurrencies  \citep{paraskevi2017btc,naimy2018btc,chu-2017}, and there is no general agreement which model is the best choice. \cite{katsiampa2017btc, naimy2018btc} and \cite{gyamerah2019btc} advocate the use of the component GARCH (CGARCH), the exponential GARCH (EGARCH) and the threshhold GARCH (TGARCH) model, respectively. \cite{chu-2017} conclude that the standard GARCH (SGARCH), the integrated GARCH (IGARCH) and the Glosten-Jagannathan-Runkle GARCH (GJRGARCH) are preferable depending on the crypto-currency. An overview over this and related literature can be found in \cite{naimy2021crypto}, who favor CGARCH, GJR-GARCH, APARCH, and TGARCH.
\cite{catania2018gas} advocate the use of a score driven model with conditional generalized hyperbolic skew student's t innovations for predicting the conditional volatility.

In this section, we consider log returns of three cryptocurrencies, namely Bitcoin (BTC), already used in subsection \ref{sec-examples}, Ethereum (ETH) and Ripple (XRP).
According to coinmarketcap.com, BTC and ETH are the cryptocurrencies with the highest and second highest  market capitalization; XRP is number seven in the list. BTC is one of the oldest cryptocurrencies, existing since 2008, whereas ETH and XRP were released in 2013 and 2012. All three cryptocurrencies are traded in many cryptocurrency exchanges like Binance, FTX or Bitstamp.
We use hourly observations from May 16, 2018 to October 27, 2021, with sample size 30264, which corresponds to 1261 days. All prices are closing values in U.S. dollars of the Bitstamp Exchange obtained from cryptodatadownload.com. Returns are estimated by taking logarithmic differences.
Figures A.4 and A.5 in the Supplementary Material \citep{holzmann2022proxy} show plots of the log-returns and the autocorrelation functions of log-returns for the three cryptocurrencies.

Similarly as in Section \ref{sec-sims}, we consider a standard GARCH(1,1) model and ARCH models of order 1 and 4. As more sophisticated models, we also employ the EGARCH model of \cite{nelson1991egarch} and the CGARCH model of  \cite{lee1999cgarch}, both with $p=q=1$ and normal innovations. We model the conditional mean by a constant value in all cases. Since the variance is non-elicitable, we aim at predicting the conditional second moment, using either squared returns or the high frequency proxy $\sum_{i=1}^{24}r_{t,i}^2$. One-step-ahead forecasts use a moving time window with length $\lfloor T/3\rfloor =420$, refitted every time step.
Figure \ref{fig:DM-BTC} in subsection \ref{sec-examples} and Figures \ref{fig:DM-ETH} and \ref{fig:DM-XRP} show the results of DM tests for the log returns of Bitcoin, Ethereum and Ripple, respectively.


\begin{figure}
\centering
\includegraphics[width=\textwidth,height=0.4\textheight]{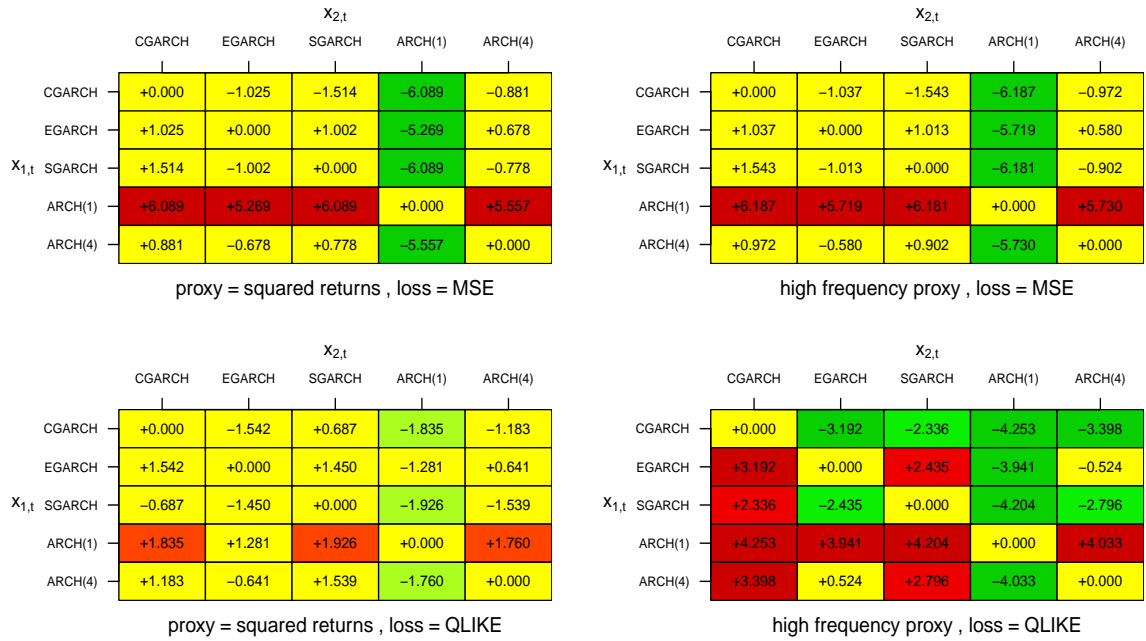}
\caption{Results of DM tests for ETH log returns from May\,16,\,2018 to October\,27,\,2021, left: squared returns, right: high frequency proxy.}
\label{fig:DM-ETH}
\end{figure}

\begin{figure}
\centering
\includegraphics[width=\textwidth,height=0.4\textheight]{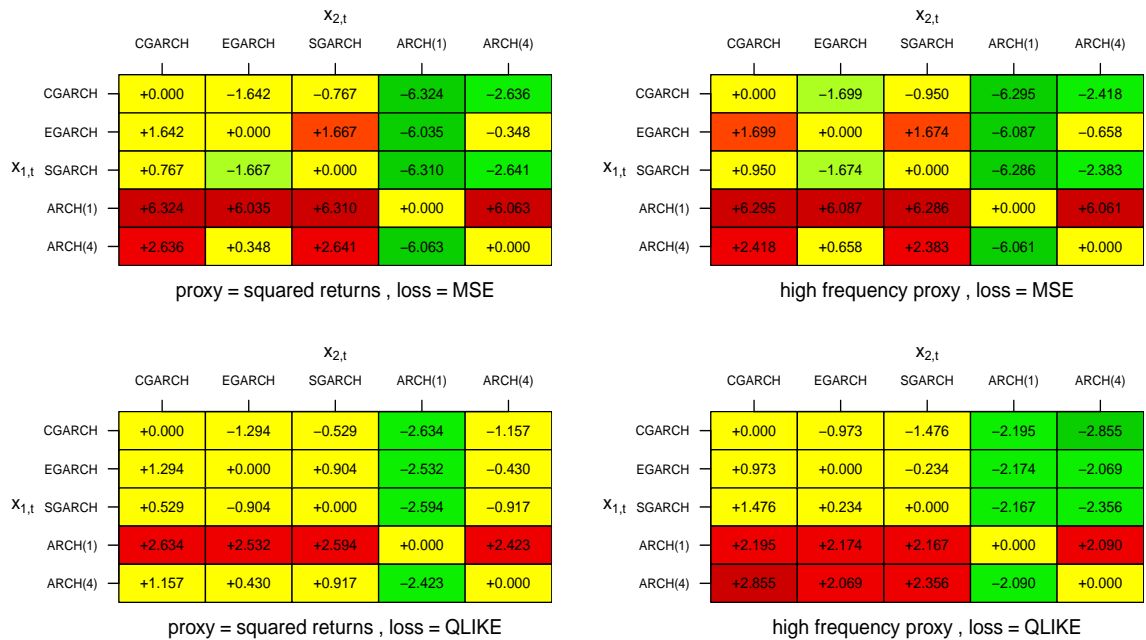}
\caption{Results of DM tests for XRP log returns from May\,16,\,2018 to October\,27,\,2021, left: squared returns, right: high frequency proxy.}
\label{fig:DM-XRP}
\end{figure}

The general observation is that the GARCH-type models outperform the two ARCH models. The GARCH-type processes are comparable, with the CGARCH dominating the EGARCH model for BTC in case of the QLIKE loss and high frequency proxy. For QLIKE loss and high frequency proxy, CGARCH outperforms all other models, followed by SGARCH.

With QLIKE loss, the power of the DM tests is considerably higher by using the high frequency proxy compared to squared returns.
Moreover, for BTC and ETH,  there is an increase in power using the QLIKE loss function compared to MSE; somewhat surprisingly, the reverse holds for XRP.
Hence, apart from the latter observation, the results corroborate the findings of the simulations.

\section{Concluding remarks}

Our perspective, say in the Theorem \ref{th:dommoment}, is that $Y_{t+1}$ is observed at time $t+1$, but that in addition a proxy $\tilde Y_{t+1}$ is also observed or can be computed which is more informative about the functional. Hence, $\tilde Y_{t+1}$ can be used to increase the power of tests for forecast dominance.
The theoretical results are supported by both the simulations and the real data example: The use of high-frequency data for the proxies generally improves the forecast evaluation. This effect overlaps with the impact of the choice
of the loss function, which can also be quite high. Empirically, the power of the DM test improves when using the QLIKE
loss compared to the MSE loss function, although no theoretical results seem to be available in this direction.

Another perspective which is pursued e.g.~in \citet{hoga2022testing} for mean forecasts of US GDP growth, by \citet{li2018asymptotic} for forecasting integrated volatility in high-frequency finance and by \citet{kleen2021measurement} for probabilistic forecasts is that the variable of interest $Y_{t+1}$ is actually not observed but latent, and only proxies of $Y_{t+1}$ with additional measurement error  can be observed. The problem is to quantify the effect of measurement error on forecast evaluation. An interesting issue in the context of probabilistic forecasting would thus be to investigate whether proxies can also be used, as in our setting, to improve forecast evaluation.

\smallskip

When moving beyond moments and ratios of moments and when considering general functionals $T$, one always has the following: If $\los$ is strictly consistent for $T$, and $\tilde Y_{t+1}$ is a proxy for $T$ of $Y_{t+1}$ in the sense that the conditional functionals
\begin{equation}\label{eq:truecondpar}
	 T\big( F_{ Y_{t+1}|\F_t}(\omega, \cdot)\big) = T\big( F_{ \tilde Y_{t+1}|\F_t}(\omega, \cdot)\big)
\end{equation}	
coincide, then for $\F_t$-measurable $Z_t$,
\begin{equation}\label{eq:condproxtrue}
	\E\, \Big[\los\big (T\big( F_{ Y_{t+1}|\F_t}(\omega, \cdot)\big), \tilde Y_{t+1}\big) | \F_t\Big](\omega) \leq \E\, \big[\los(Z_t, \tilde Y_{t+1}) | \F_t\big](\omega) \qquad \text{ for }\Pb-a.e.\ \omega \in \Omega,
\end{equation}
with equality almost surely if and only if $T\big( F_{ Y_{t+1}|\F_t}(\omega, \cdot)\big) = Z_t$ $\Pb$-almost surely.
Indeed, this is simply \eqref{eq:ineqfirst} stated for $\tilde Y_{t+1}$ by observing \eqref{eq:truecondpar}.
\eqref{eq:condproxtrue} implies that if we do not take into account comparing \emph{two} possibly misspecified forecasts, then we can always replace  the variable of interest $Y_{t+1}$ by a proxy $\tilde Y_{t+1}$ which satisfies \eqref{eq:truecondpar}.

However, \citet{hoga2022testing} show in their Proposition 3 that for quantile scores, this cannot be extended to comparing misspecified forecasts: If the difference of conditional loss differences vanishes, then the conditional distributions of  $Y_{t+1}$ and $\tilde Y_{t+1}$ must coincide. It would be of interest to investigate if this negative result is more pervasive and applies to further functionals such as expectiles.

\section*{Acknowledgements}
We would like to thank Andrew Patton and Tilmann Gneiting for pointers to the literature, and in particular for bringing the paper by \citet{hoga2022testing} to our attention. Then, we would like to thank Timo Dimitriadis for providing various general and detailed, helpful comments on the paper.

The authors are grateful for the suggestions and comments of two anonymous reviewers  which helped to improve the paper.

%
%
\bibliographystyle{apalike}

\begin{thebibliography}{}

\bibitem[Amaya et~al., 2015]{amaya2015realized}
Amaya, D., Christoffersen, P., Jacobs, K., and Vasquez, A. (2015).
\newblock Does realized skewness predict the cross-section of equity returns?
\newblock {\em Journal of Financial Economics}, 118:135--167.

\bibitem[Barndorff-Nielsen, 1997]{barndorff1997nig}
Barndorff-Nielsen, O. (1997).
\newblock Normal inverse gaussian distributions and stochastic volatility
  modelling.
\newblock {\em Scandinavian Journal of Statistics}, 24:1--13.

\bibitem[Brooks et~al., 2005]{brooks2005autoregressive}
Brooks, C., Burke, S.~P., Heravi, S., and Persand, G. (2005).
\newblock Autoregressive conditional kurtosis.
\newblock {\em Journal of Financial Econometrics}, 3:399--421.

\bibitem[Catania et~al., 2018]{catania2018gas}
Catania, L., Grassi, S., Ravazzolo, F. (2018).
\newblock Predicting the volatility of cryptocurrency time-series.
\newblock {In: \em Corazza M., Durb\'{a}n M., Gran\'{e} A., Perna C., Sibillo M. (eds) Mathematical and Statistical Methods for Actuarial Sciences and Finance}, Springer.

\bibitem[Chu et~al., 2017]{chu-2017}
Chu, J., Chan, S., Nadarajah, S., and Osterrieder, J. (2017).
\newblock Garch modelling of cryptocurrencies.
\newblock {\em Journal of Risk and Financial Management}, 10(4).

\bibitem[Corsi, 2009]{corsi2009simple}
Corsi, F. (2009).
\newblock A simple approximate long-memory model of realized volatility.
\newblock {\em Journal of Financial Econometrics}, 7(2):174--196.

\bibitem[Diebold, 2015]{diebold2015comparing}
Diebold, F.~X. (2015).
\newblock Comparing predictive accuracy, twenty years later: A personal
  perspective on the use and abuse of diebold-mariano tests.
\newblock {\em Journal of Business and Economic Statistics}, 33:1--24.

\bibitem[Diebold and Mariano, 1995]{diebold1995comparing}
Diebold, F.~X. and Mariano, R.~S. (1995).
\newblock Comparing predictive accuracy.
\newblock {\em Journal of Business and Economic Statistics}, 13:253--263.

\bibitem[Diks et~al., 2011]{diks2011likelihoodbased}
Diks, C., Panchenko, V., and van Dijk, D. (2011).
\newblock Likelihoodbased scoring rules for comparing density forecasts in
  tails.
\newblock {\em J. Econometrics}, 163:215--230.

\bibitem[Ding et~al., 1993]{ding1993aparch}
Ding, Z., Granger, C., and Engle, R. (1993).
\newblock A long memory property of stock market returns and a new model.
\newblock {\em Journal of Empirical Finance}, 83:83--106.

\bibitem[Ehm et~al., 2016]{ehm2016quantiles}
Ehm, W., Gneiting, T.,  Jordan, A., and Kr{\"u}ger, F.  (2016).
\newblock Of quantiles and expectiles: consistent scoring functions, Choquet representations and forecast rankings.
\newblock {\em Journal of the Royal Statistical Society: Series B (Statistical Methodology)}, 78:505--562.

\bibitem[Fissler et~al., 2016]{fissler2016expected}
Fissler, T., Ziegel, J.~F., and Gneiting, T. (2016).
\newblock Expected shortfall is jointly elicitable with value at risk -
  implications for backtesting.
\newblock {\em Risk Magazine}, January:58--61.

\bibitem[Ghalanos, 2020]{ghalanos2020rugarch}
Ghalanos, A. (2020).
\newblock {\em rugarch: Univariate GARCH models.}
\newblock R package version 1.4-4.

\bibitem[Giacomini and White, 2006]{giacomini2006tests}
Giacomini, R. and White, H. (2006).
\newblock Tests of conditional predictive ability.
\newblock {\em Econometrica}, 74:1545--1578.

\bibitem[Gneiting, 2011]{gneiting2011making}
Gneiting, T. (2011).
\newblock Making and evaluating point forecasts.
\newblock {\em Journal of the American Statistical Association},
  106(494):746--762.

\bibitem[Gneiting and Ranjan, 2011]{gneiting2011comparing}
Gneiting, T. and Ranjan, R. (2011).
\newblock Comparing density forecasts using threshold- and quantile-weighted
  scoring rules.
\newblock {\em Journal of Business and Economic Statistics}, 29:411--422.

\bibitem[Gyamerah, 2019]{gyamerah2019btc}
Gyamerah, S. (2019).
\newblock Modelling the volatility of Bitcoin returns using GARCH models.
\newblock {\em Quantitative Finance and Economics}, 3:739--53.

\bibitem[Hansen and Lunde, 2006]{hansen2006consistent}
Hansen, P.~R. and Lunde, A. (2006).
\newblock Consistent ranking of volatility models.
\newblock {\em Journal of Econometrics}, 131(1-2):97--121.

\bibitem[Harvey and Siddique, 1999]{harvey1999autoregressive}
Harvey, C.~R. and Siddique, A. (1999).
\newblock Autoregressive conditional skewness.
\newblock {\em The Journal of Financial and Quantitative Analysis},
  34:465--487.

\bibitem[Harvey and Siddique, 2000]{harvey2000conditional}
Harvey, C.~R. and Siddique, A. (2000).
\newblock Conditional skewness in asset pricing tests.
\newblock {\em Journal of Finance}, 55:1263--1295.

\bibitem[Hoga and Dimitriadis, 2022]{hoga2022testing}
Hoga, Y. and Dimitriadis, T. (2022).
\newblock On testing equal conditional predictive ability under measurement error.
\newblock {\em Journal of Business \& Economic Statistics}, 0:1--13.


\bibitem[Holzmann and Eulert, 2014]{holzmann2014role}
Holzmann, H., and Eulert, M. (2014).
\newblock The role of the information set for forecasting - with applications to risk management.
\newblock {\em Annals of Applied Statistics}, 8:595--621.

\bibitem[Holzmann and Klar, 2022]{holzmann2022proxy}
Holzmann, H., and Klar, B. (2022).
\newblock Supplement to ``Using proxies to improve forecast evaluation.''
\newblock {\em Annals of Applied Statistics}, 8:595--621.

\bibitem[Katsiampa, 2017]{katsiampa2017btc}
Katsiampa, P. (2017).
\newblock Volatility estimation for Bitcoin: A comparison of GARCH models.
\newblock {\em Economics Letters}, 158:3--6.

\bibitem[Kleen, 2021]{kleen2021measurement}
Kleen, O. (2021).
\newblock Measurement error sensitivity of loss functions for distribution forecasts.
\newblock {\em Available at SSRN 3476461}.

\bibitem[Koopman et~al., 2005]{koopman2005forecasting}
Koopman, S.~J., Jungbacker, B., and Hol, E. (2005).
\newblock Forecasting daily variability of the s\&p 100 stock index using
  historical, realised and implied volatility measurements.
\newblock {\em Journal of Empirical Finance}, 12(3):445--475.


\bibitem[Lambert and Laurent, 2002]{lambert2002modeling}
Lambert, P. and Laurent, S. (2002).
\newblock Modeling skewness dynamics in series of financial data using skewed
  location-scale distributions.
\newblock {\em Working Paper, Universit\'e Catholique de Louvain and
  Universit\'e de Li\`{e}ge.}

\bibitem[Lau, 2015]{lau2015simple}
Lau, C. (2015).
\newblock A simple normal inverse gaussian-type approach to calculate
  value-at-risk based on realized moments.
\newblock {\em Journal of Risk}, 17:1--18.

\bibitem[Laurent et~al., 2013]{laurent2013loss}
Laurent, S., Rombouts, J.~V., and Violante, F. (2013).
\newblock On loss functions and ranking forecasting performances of
  multivariate volatility models.
\newblock {\em Journal of Econometrics}, 173(1):1--10.

\bibitem[Lee and Engle (1999)]{lee1999cgarch}
Lee, G.J., and Engle, R.F. (1999).
\newblock A permanent and transitory component model of stock return volatility.
\newblock {In \em Cointegration, Causality and Forecasting: A Festschrift in Honor of Clive W.J. Granger.}, 475-497.

\bibitem[Lerch et~al., 2017]{lerch2017forecaster}
Lerch, S., Thorarinsdottir, T.~L., Ravazzolo, F., and Gneiting, T. (2017).
\newblock Forecaster's dilemma: Extreme events and forecast evaluation.
\newblock {\em Statist. Sci.}, 32:106--127.

\bibitem[Li and Patton, 2018]{li2018asymptotic}
Li, J. and Patton, A.~J. (2018).
\newblock Asymptotic inference about predictive accuracy using high frequency
  data.
\newblock {\em Journal of Econometrics}, 203(2):223--240.

\bibitem[Naimy and Hayek, 2018]{naimy2018btc}
Naimy, V. and Hayek, M. (2018).
\newblock Modelling and predicting the bitcoin volatility using garch models.
\newblock {\em Int. J. Mathematical Modelling and Numerical Optimisation},
  8:197--215.

\bibitem[Naimy et~al., 2021]{naimy2021crypto}
Naimy, V., Haddad, O., Fern\'{a}ndez-Avil\'{e}s, G. and El Khoury, R. (2021).
\newblock The predictive capacity of GARCH-type models in measuring the volatility of crypto and world currencies.
\newblock {\em PLoS ONE}, 16(1):e0245904.

\bibitem[Nelson, 1991]{nelson1991egarch}
Nelson, D.~B. (1991).
\newblock Conditional heteroskedasticity in asset returns: A new approach.
\newblock {\em Econometrica}, 59(2):347--370.

\bibitem[Neuberger, 2012]{neuberger2012realized}
Neuberger, A. (2012).
\newblock Realized skewness.
\newblock {\em Review of Financial Studies}, 25:3423--3455.

\bibitem[Nolde and Ziegel, 2017]{nolde2017elicitability}
Nolde, N. and Ziegel, J.~F. (2017).
\newblock Elicitability and backtesting: Perspectives for banking regulation.
\newblock {\em The annals of applied statistics}, 11:1833--1874.

\bibitem[Paraskevi, 2017]{paraskevi2017btc}
Paraskevi, K. (2017).
\newblock Volatility estimation for bitcoin: A comparison of garch models.
\newblock {\em Economics Letters}, 158:3--6.

\bibitem[Patton, 2009]{patton2009evaluating}
Patton, A.~J., and Sheppard, K. (2009).
\newblock Evaluating volatility and correlation forecasts.
\newblock in: {\em  Handbook of financial time series}, 801--838, Springer.

\bibitem[Patton, 2011]{patton2011volatility}
Patton, A.~J. (2011).
\newblock Volatility forecast comparison using imperfect volatility proxies.
\newblock {\em Journal of Econometrics}, 160:246--256.

\bibitem[Patton, 2020]{patton2020comparing}
Patton, A.~J. (2020).
\newblock Comparing possibly misspecified forecasts.
\newblock {\em Journal of Business \& Economic Statistics}, 38(4):796--809.

\bibitem[{R Core Team}, 2021]{rcore2021r}
{R Core Team} (2021).
\newblock {\em R: A Language and Environment for Statistical Computing}.
\newblock R Foundation for Statistical Computing, Vienna, Austria.

\bibitem[Savage, 1971]{savage1971elicitation}
Savage, L.~J. (1971).
\newblock Elicitation of personal probabilities and expectations.
\newblock {\em Journal of the American Statistical Association},
  66(336):783--801.

\bibitem[Shen et~al., 2018]{shen2018surprising}
Shen, K., Yao, J., and Li, W.~K. (2018).
\newblock On the surprising explanatory power of higher realized moments in
  practice.
\newblock {\em Statistics and Its Interface}, 11:153--168.

\bibitem[Steinwart et~al., 2014]{steinwart2014elicitation}
Steinwart, I., Pasin, C., Williamson, R., and Zhang, S. (2014).
\newblock Elicitation and identification of properties.
\newblock In {\em Conference on Learning Theory}, pages 482--526. PMLR.

\bibitem[Wuertz et~al., 2020]{wuertz2020fgarch}
Wuertz, D., Setz, T., Chalabi, Y., Boudt, C., Chausse, P., and Miklovac, M.
  (2020).
\newblock {\em fGarch: Rmetrics - Autoregressive Conditional Heteroskedastic
  Modelling}.
\newblock R package version 3042.83.2.
\end{thebibliography}

\pagebreak
\section{Supplementary Material}
\renewcommand{\floatpagefraction}{0.9}
\renewcommand{\textfraction}{0.1}

\subsection{Additional simulation results}

In addition to the results in Figures 1 and 2 in Section 6.1
for realized volatilities, we also replaced these by the adjusted intra-daily log range, given by
\begin{align}
	\left(\max_s \log P_s - \min_s \log P_s\right)/(2\sqrt{\log 2}), \quad t-1 < s \leq t,
\end{align}
where $(P_s)$ denotes the price process. This volatility proxy is unbiased under normality assumptions; details can be found in Patton (2011), 
p.~250. The results are displayed in  Figure \ref{fig3}.
All in all, forecast evaluation using the adjusted intra-daily log range is less sharp compared to the realized volatilities with $m=13$, but power is clearly higher than using squared returns as proxies.

\begin{figure}[hb]
	\centering
\includegraphics[width=0.9\textwidth,height=0.4\textheight]
	{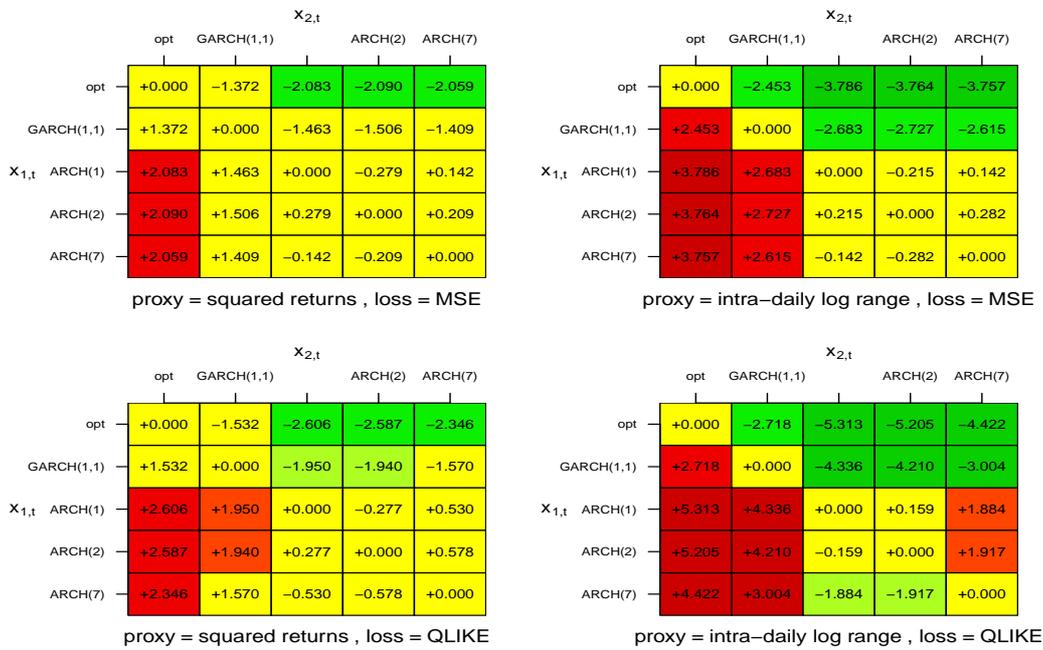}
	\caption{Results of Diebold-Mariano tests, normal distribution, left: squared returns, right: adjusted intra-daily log range, $m=100, T=1500$.}
	\label{fig3}
\end{figure}

\bigskip

The left and right panels of Fig.~\ref{fig12} show the results of the DM tests in the setting of Section 6.2.2, 
using $r_t^4$ and the realized corrected fourth moment as proxies, respectively, when replacing  normal by nig - distributed innovations. The power of the DM test decreases strongly. On the other hand, the entries are somewhat larger as in forecasting the third moment (with $T=1500$). Here, at least a few values are significant on the 0.1-level.

\begin{figure}
	\centering
\includegraphics[width=0.9\textwidth,height=0.4\textheight]
	{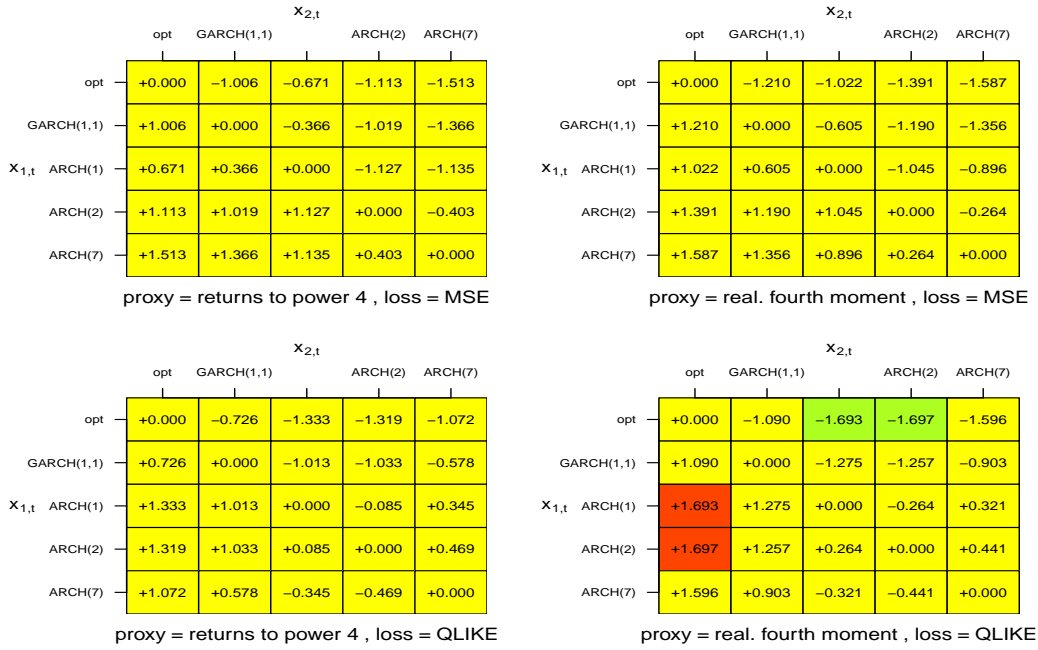}
	\caption{Results of Diebold-Mariano tests, nig-distribution, left: returns to the power 4, right: realized corrected 4th moment, m=13, T=1500.}
	\label{fig12}
\end{figure}

Finally, we consider again volatility forecasts, but now based on the apARCH process with exponent 4 from Section 6.3. 
The results are shown in Figure \ref{fig14}.
We see a slight increase in power compared to Fig. 6; 
similar as for the GARCH process, differentiating between volatility forecasts is easier compared to forecasts of the 4th moment in the case of the apARCH process at hand.

\begin{figure}
\centering
\includegraphics[width=0.9\textwidth,height=0.4\textheight]
{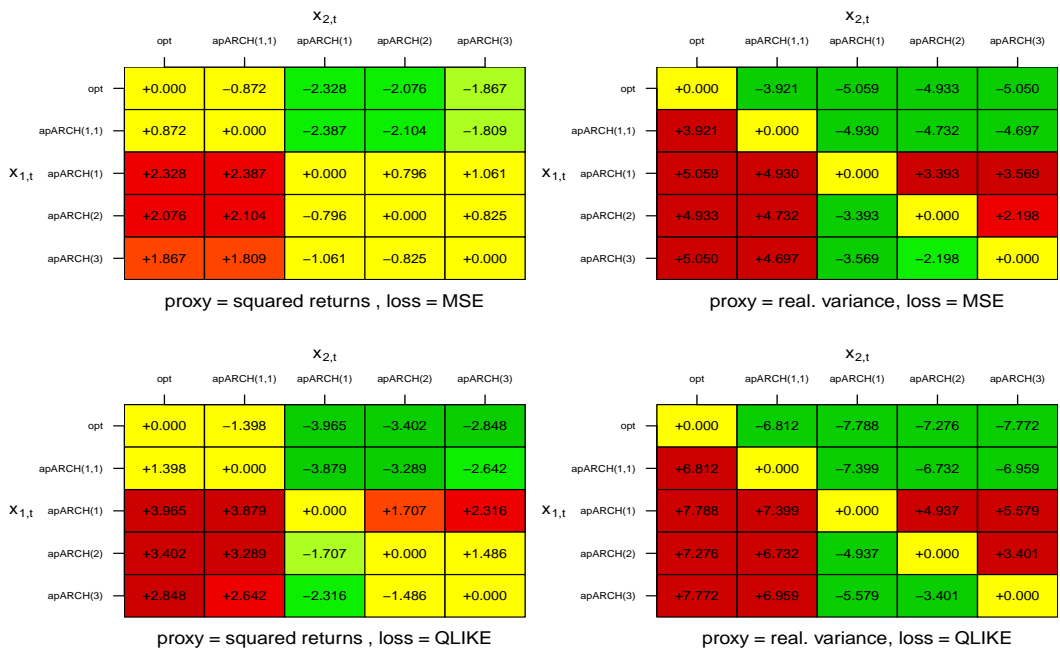}
\caption{Results of Diebold-Mariano tests, apARCH process with exponent 4, normal distribution, left: returns to the power 2, right: realized 2nd moment, m=100, T=1500.}
\label{fig14}
\end{figure}

\subsection{Additional material for the real data analysis}

Figures \ref{fig-tsplot} and  \ref{fig15} show plots of the log-returns and the autocorrelation functions of log-returns for BTC, ETH and XRP. The autocorrelations are quite small for all three cryptocurrencies; the lag-one autocorrelations are negativ.

\begin{figure}
\centering
\includegraphics[width=0.9\textwidth,height=0.29\textheight]
{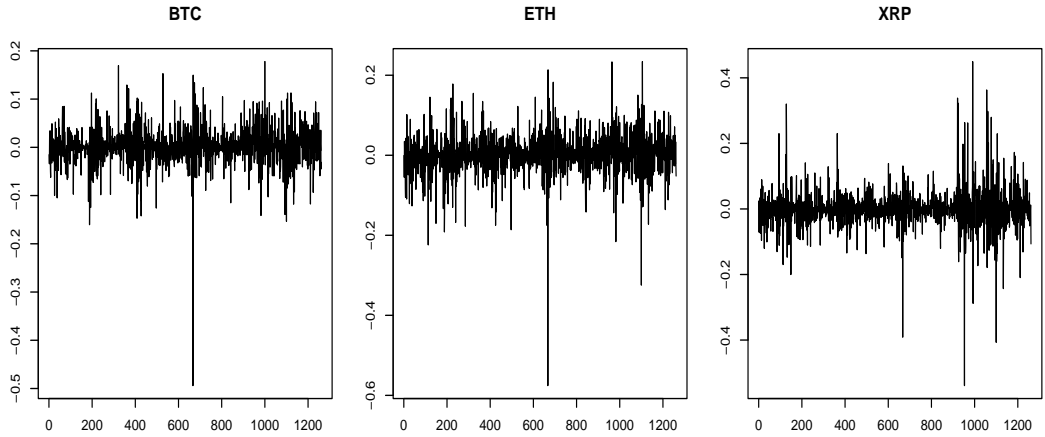}
\caption{Plots of daily cryptocurrency log-returns from May 16, 2018 to October 27, 2021}
\label{fig-tsplot}
\end{figure}

\begin{figure}
\centering
\includegraphics[width=0.9\textwidth,height=0.29\textheight]
{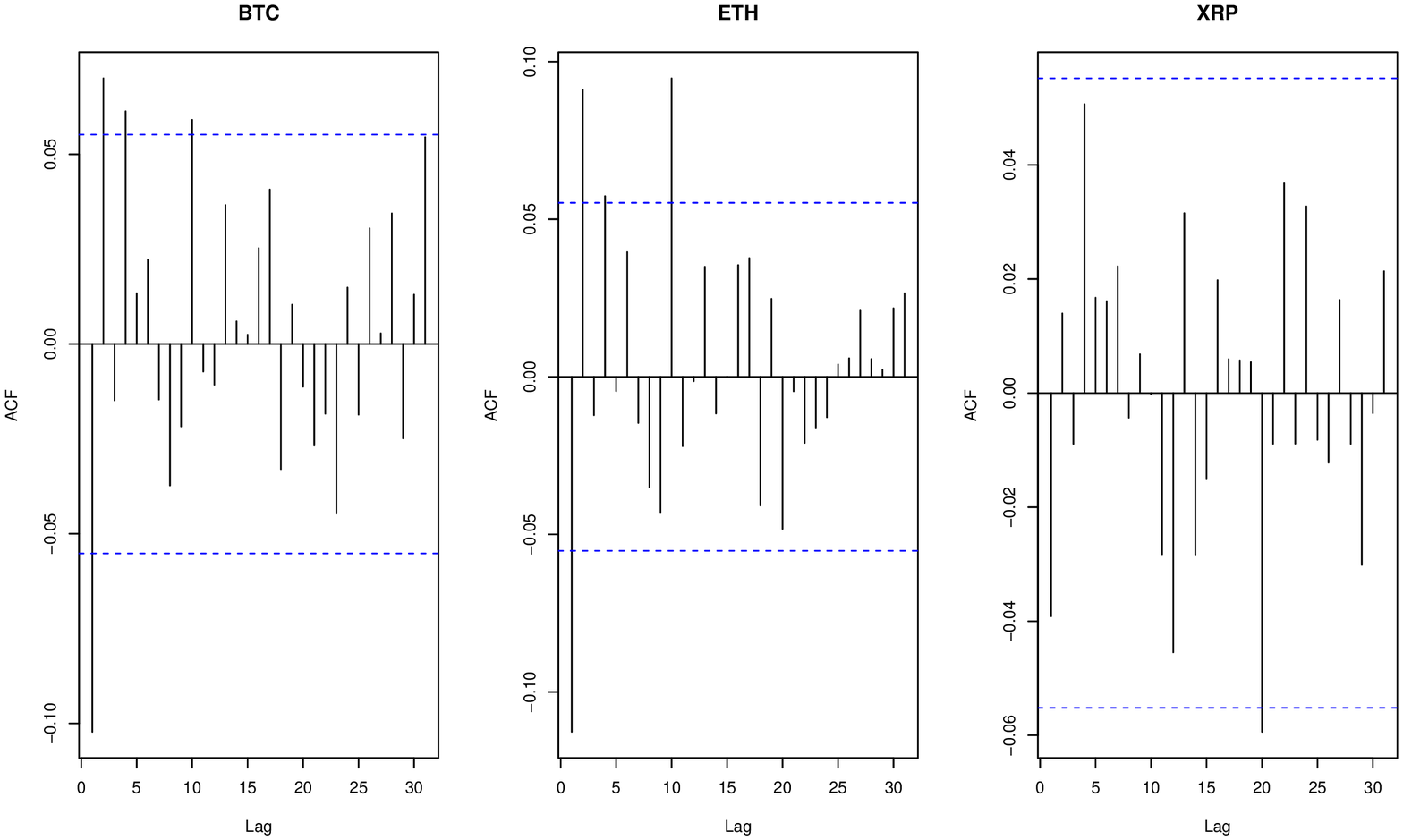}
\caption{Autocorrelation functions of daily cryptocurrency log-returns}
\label{fig15}
\end{figure}

However, as Figure \ref{fig16} reveals, the autocorrelations of the absolute values of the log-returns are much more pronounced, and they are persistent over several days or even weeks.

\begin{figure}
\centering
\includegraphics[width=0.95\textwidth,height=0.29\textheight]
{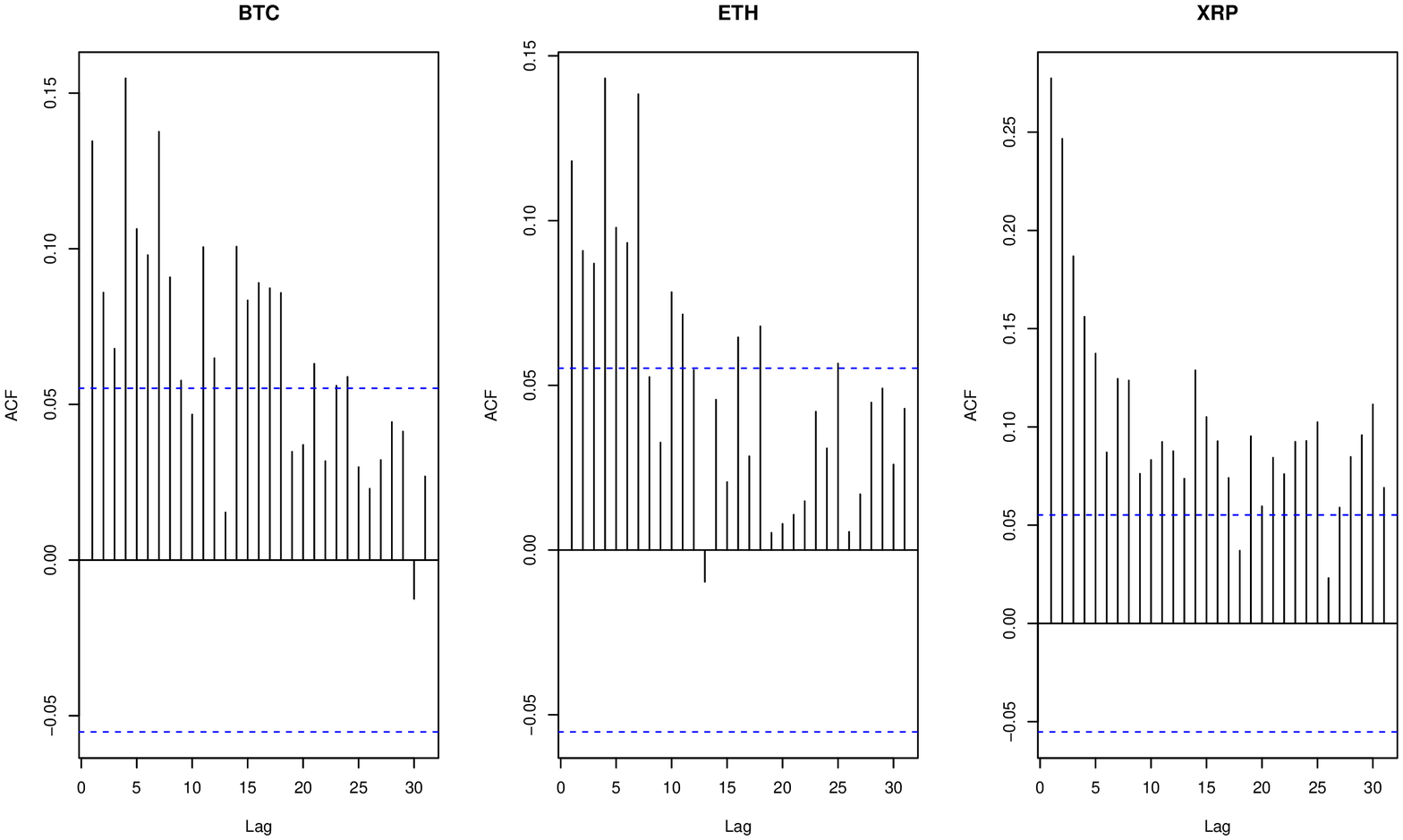}
\caption{Autocorrelation functions of absolute values of daily cryptocurrency log-returns}
\label{fig16}
\end{figure}

\end{document}